\documentclass[sigconf, nonacm]{acmart}

\usepackage{amsmath, amsfonts}
\usepackage{amssymb}
\usepackage{graphicx}
\usepackage{textcomp}
\usepackage[dvipsnames]{xcolor}
\usepackage{subcaption}
\usepackage{xspace}
\usepackage{tikz} 
\usepackage{bbding} 
\usepackage{diagbox} 
\usepackage[normalem]{ulem}

\usepackage{pifont}

\usepackage{algorithm}
\usepackage{algpseudocodex}
\usepackage{algorithmicx}

\usepackage{environ}
\usepackage{adjustbox}

\usepackage{pgfplots}
\usetikzlibrary{patterns}
\usepgfplotslibrary{groupplots}

\usepackage{todonotes}

\usepackage{enumitem}

\usepackage[most]{tcolorbox}

\usepackage{wrapfig}

\usepackage{amsthm}  

\usetikzlibrary{spy}

\usepackage{tabularx}
\usepackage{soul}

\usepackage{multirow}

\usepackage{capt-of}

\usepackage{mathpartir}
\usepackage{mathtools}

\usepackage{marginnote}

\newcommand{\green}[1]{\textcolor{green}{#1}}

\newcommand{\brown}[1]{\textcolor{brown}{#1}}
\newcommand{\orange}[1]{\textcolor{orange}{#1}}
\newcommand{\purple}[1]{\textcolor{purple}{#1}}
\newcommand{\violet}[1]{\textcolor{violet}{#1}}

\newcommand{\set}[1]{\{#1\}}

\newcommand{\dom}{\textsf{dom}}

\newcommand{\bfit}[1]{\emph{#1}}

\renewcommand{\O}{\mathit{O}}
\renewcommand{\H}{\mathcal{H}}
\newcommand{\h}{\mathit{h}}

\newcommand{\Key}{\mathsf{K}}
\newcommand{\Val}{\mathsf{V}}
\newcommand{\init}{\mathsf{init}}
\newcommand{\Op}{\mathsf{Op}}

\newcommand{\OpId}{\mathsf{OpId}}
\newcommand{\opid}{\iota}
\newcommand\R{\mathsf{R}}
\newcommand\W{\mathsf{W}}
\newcommand{\WriteTx}{\mathsf{WriteTx}}

\renewcommand\AE{\mathcal{X}}

\newcommand{\aecav}{\xi}
\newcommand{\rel}[1]{\xrightarrow{#1}}

\newcommand{\comp}{\;;\;}
\newcommand{\po}{\textup{\textsf{\purple{po}}}}
\newcommand{\SO}{\textup{\textsf{\purple{SO}}}}
\newcommand{\VIS}{\textup{\textsf{\brown{VIS}}}}

\newcommand{\nVIS}{\overline{\VIS}}
\newcommand{\AR}{\textup{\textsf{\orange{AR}}}}
\newcommand{\CO}{\textup{\textsf{\orange{CO}}}}

\newcommand{\Int}{\textup{\textsc{\violet{Int}}}}
\newcommand{\Ext}{\textup{\textsc{\violet{Ext}}}}
\newcommand{\Session}{\textup{\textsc{\violet{Session}}}}
\newcommand{\TransVis}{\textup{\textsc{\violet{TransVis}}}}
\newcommand{\Prefixaxiom}{\textup{\textsc{\violet{Prefix}}}}
\newcommand{\TotalVis}{\textup{\textsc{\violet{TotalVis}}}}
\newcommand{\wwconflict}{\bowtie}
\newcommand{\rwconflict}{\triangleleft}
\newcommand{\NoConflict}{\textup{\textsc{\violet{NoConflict}}}}

\newcommand{\il}{\ell}

\newcommand{\IL}{\mathcal{L}}

\newcommand{\RA}{\textup{\textsf{RA}}}
\newcommand{\CC}{\textup{\textsf{CC}}}

\newcommand{\PC}{\textup{\textsf{PC}}}
\newcommand{\PSI}{\textup{\textsf{PSI}}}
\newcommand{\SI}{\textup{\textsf{SI}}}

\newcommand{\SER}{\textup{\textsf{SER}}}

\newcommand{\allLevelsSet}{\set{\RA,\CC,\PC,\PSI,\SI,\SER}}

\newcommand{\ReadAtomic}{\textup{\textsf{ReadAtomic}}}
\newcommand{\Causal}{\textup{\textsf{Causal}}}
\newcommand{\Prefix}{\textup{\textsf{Prefix}}}
\newcommand{\Conflict}{\textup{\textsf{Conflict}}}
\newcommand{\SnapshotIsolation}{\textup{\textsf{SnapshotIsolation}}}
\newcommand{\Serializability}{\textup{\textsf{Serializability}}}

\newcommand{\T}{\mathcal{T}}

\newcommand{\WR}{\green{\textup{\textsf{WR}}}}

\newcommand{\hStatex}[0]{\vspace{4pt}}

\newcommand{\store}{\textsf{store}}
\newcommand{\buffer}{\mathit{buffer}}
\newcommand{\level}{\mathsf{level}}
\newcommand{\StartProc}{\textsc{Start}}

\newcommand{\WriteProc}{\textsc{Write}}
\newcommand{\ReadProc}{\textsc{Read}}
\newcommand{\CommitProc}{\textsc{Commit}}

\newcommand{\sts}{\mathit{sts}}
\newcommand{\cts}{\mathit{cts}}
\newcommand{\xlock}{\mathit{xlock}}
\newcommand{\slock}{\mathit{slock}}
\newcommand{\releaseLocks}{\mathit{unlocks}}

\newcommand{\committed}{\textsf{committed}}
\newcommand{\aborted}{\textsf{aborted}}

\newcommand{\OneWrite}{\textsf{OneWrite}}
\newcommand{\RYW}{\textsf{RYW}}
\newcommand{\InitTran}{\textsf{InitTran}}

\newcommand{\code}[2]{line~#1:#2}

\newcommand{\TS}{\textsf{T}}

\newcommand{\now}{\textsf{now()}}
\newcommand{\partialmapsto}{\rightharpoonup}

\algdef{SE}[DOWHILE]{Indent}{EndIndent}{}{}

\newcommand{\case}{\textsc{Case}}
\newcommand{\casei}{\textsc{Case I}}
\newcommand{\caseii}{\textsc{Case II}}
\newcommand{\caseiii}{\textsc{Case III}}

\newcommand{\inlsec}[1]{\smallskip\noindent\textbf{#1.}}

\newcommand{\ourframework}{\textsc{MixIso}\xspace}

\begin{document}

\title{Semantic Conformance of Concurrency Control Protocols under Mixed Isolation Levels}
\titlenote{This work was accepted by and will be presented at 1st Symposium on Consistency Checking Principles (SCCP 2026).}

\author{Qiuhuan Xiong}
\affiliation{%
  \institution{Nanjing University}
  \country{}
}
\email{qiuhuanxiong@smail.nju.edu.cn}

\author{Hengfeng Wei}
\affiliation{%
  \institution{Hunan University}
  \country{}
}
\email{hfwei@hnu.edu.cn}

\author{Si Liu}
\affiliation{%
  \institution{Texas A\&M University}
  \country{}
}
\email{si.liu@tamu.edu}

\author{Yuxing Chen}
\affiliation{%
  \institution{Renmin University of China}
  \country{}
}
\email{axinggu@gmail.com}

\author{Jidong Ge}
\affiliation{%
  \institution{Nanjing University}
  \country{}
}
\email{gjd@nju.edu.cn}

\begin{abstract}

Modern database systems widely support per-transaction isolation levels
as a practical means of balancing consistency guarantees and performance.
Yet, it remains largely unclear whether their concurrency control protocols  correctly enforce the intended isolation guarantees under such mixed-isolation settings.
In this paper, we address this semantic conformance question by
developing \ourframework, a formal semantic framework for mixed isolation levels.
We demonstrate  its applicability  
by establishing the semantic conformance of two concurrency control protocols,
 one combining two isolation levels and the other three.


\end{abstract}

\maketitle

\section{Introduction }
\label{sec-intro}

 Strong database isolation guarantees, such as serializability (\SER)~\cite{SER:JACM1979}, 
often come at the cost of increased concurrency control overheads.
To mitigate this tension,
production database systems,
including Oracle, MySQL, PostgreSQL, CockroachDB, and TiDB, 
increasingly support assigning isolation levels on a 
\emph{per-transaction} basis.
This allows applications to execute certain transactions under weaker isolation levels
like \emph{snapshot isolation} (\SI)~\cite{Critique:SIGMOD1995},
instead of enforcing \SER{} system-wide,
thereby achieving performance gains.

However, this raises a natural question:
\emph{when transactions execute under such mixed isolation settings,
can the underlying concurrency control mechanisms correctly enforce the desired isolation guarantees?} 
This question remains largely unanswered despite
prior efforts~\cite{DBLP:journals/fac/LiuOWGM19,10.1007/978-3-031-90660-2_3,10.1145/3494517,10.1145/3269981,VerIso:VLDB2025,ConsMaude:TACAS2019,10.14778/3750601.3750626,10.1145/3747515}
on verifying isolation guarantees for concurrency control protocols, which focus on homogeneous isolation settings.

The challenges are mainly twofold:

\begin{enumerate}[label=(\arabic*), leftmargin=20pt]
    \item
    The semantics of individual isolation guarantees and their variants are already subtle~\cite{DBLP:conf/podc/CrooksPAC17,noc-noc};
    mixing them introduces significantly greater semantic complexity.
    In fact, despite widespread support for per-transaction isolation, existing database documentation~\cite{CockroachDBTxnIsoLevel,YugabyteDBTxnIsoLevel,TiDBTxnIsoLevel,SQLServerTxnIsoLevel,OracleTxnIsoLevel,MySQLTxnIsoLevel,PostgreSQLTxnIsoLevel} typically leaves the semantics of mixing isolation levels underspecified or implementation-specific.

    \item Answering this question demands a mathematically rigorous and systematic framework, 
    ideally along with tool support,
    for formally verifying the semantic conformance of concurrency control protocols across a wide range of isolation level combinations.
\end{enumerate}

\inlsec{\ourframework}
We address these challenges
by developing 
 \ourframework,
 a formal semantic framework for mixed isolation levels.
It accommodates
both traditional  levels, such as \SER{}  and \SI, 
and a range of emerging ones, including
\emph{read atomicity} (RA)~\cite{RAMP:TODS2016},
\emph{transactional causal consistency} (CC)~\cite{Cure:ICDCS2016},
\emph{prefix consistency} (\PC)~\cite{GSP:ECOOP2015},
and
\emph{parallel snapshot isolation} (PSI)~\cite{PSI:SOSP2011}.

Underlying \ourframework is a semantic characterization of mixed isolation levels based on
\emph{per-transaction visibility}:
the effects of other transactions that a given transaction is allowed to observe in a database execution.
The overall execution is then considered 
\emph{consistent}
 if, for each transaction, the effects it observes conform to its assigned  level.
This per-transaction interpretation of visibility enables us to faithfully capture subtle semantic differences between isolation levels, 
 while providing the flexibility to
   specify mixed isolation
    guarantees that reflect  particular design choices or implementations.

\ourframework builds on the axiomatic framework in~\cite{Framework:CONCUR2015},
which provides a declarative basis for specifying isolation levels while abstracting away implementation details.
Yet, extending its axioms to mixed isolation settings is non-trivial.

A key challenge lies in
employing a shared \emph{global} visibility relation
across transactions running at different isolation levels
while preserving \emph{isolation autonomy}, 
i.e., a transaction's isolation guarantee depends only on its own assigned level. 
Otherwise, a transaction with a stronger  level could implicitly impose its visibility requirements on other transactions.

%

\begin{figure}[h!]
	\centering
	\includegraphics[width=.75\columnwidth]{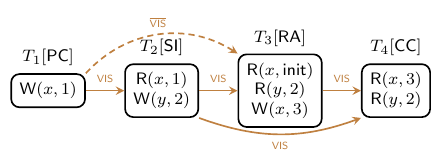}
	\captionsetup{skip=0pt}
	\caption{
		An illustration of isolation autonomy for $T_3$.
	}
	\label{fig:intro-example}
\end{figure}
For instance, as shown in  Figure~\ref{fig:intro-example},
 transaction $T_4$, which is assigned CC,
requires its visible set to be transitively closed: 
since $T_4$ sees\footnote{``See’’ is indicated by the \VIS{} relation between transactions; see Section~\ref{sec-mil}.} $T_3$ 
(by reading the value 3 of key $x$ written by $T_3$) 
and $T_3$ sees $T_2$,
 $T_4$ must also see $T_2$, which indeed holds.
Yet, $T_3$, which is assigned the weaker level RA
 and lies in the visible past of $T_4$, 
 is not thereby required to satisfy transitive visibility. 
  In particular, although $T_3$ sees $T_2$ and $T_2$ sees $T_1$, $T_3$ need not see $T_1$, 
      which is allowed under RA.
  In other words, 
 the visibility requirements imposed on $T_4$ by CC do not propagate to transactions assigned weaker isolation levels.

 
 We address this challenge
  by specifying, 
  for each transaction, 
  axioms over the set of transactions visible to it, 
   based on its assigned isolation level.

%

\inlsec{Case Studies}
Leveraging \ourframework,
we  establish
 the semantic conformance of 
 two
  concurrency control protocols:
(i) the protocol used in Microsoft SQL Server and Oracle Berkeley DB~\cite{Allocating:PODS2005,ModOrMixing:DASFAA2008}  combining \SI{} and strict two-phase locking (S2PL); 
and 
(ii) a new protocol mixing three isolation levels ($\PC$, $\SI$, and $\SER$). 
Our proofs follow prior construction-based approach~\cite{burckhardtPrinEvenCons2014,xiongECOOP2020},
and are currently carried out manually.

\section{Preliminaries}
\label{sec-models}


\textbf{Isolation Levels.}
Databases support a range of isolation levels
to accommodate different trade-offs between consistency  and performance.
These include 
\emph{read atomicity}
($\RA$)~\cite{RAMP:TODS2016},
\emph{transactional causal consistency} ($\CC$)~\cite{Cure:ICDCS2016},
\emph{prefix consistency} ($\PC$)~\cite{GSP:ECOOP2015},
\emph{parallel snapshot isolation} ($\PSI$)~\cite{PSI:SOSP2011},
\emph{snapshot isolation} ($\SI$)~\cite{Critique:SIGMOD1995},
and
 \emph{serializability} ($\SER$)~\cite{SER:JACM1979}.
Appendix~\ref{app:isolation} illustrates these isolation levels with examples.

%
%
%
%
%

\inlsec{Transactions}
We consider a transactional key-value store (KVS) managing a set of keys $\Key = \set{x,y,z,\dots}$ with values from a set $\Val$.
We denote by $\Op$ the set of read and write operations on keys:
$\Op = \set{\R_{\opid}(x,v),\;\W_{\opid}(x,v) \mid \opid \in \OpId,\; x \in \Key,\; v \in \Val}$,
where $\OpId$ is the set of operation identifiers (omitted when unimportant).
	A \emph{transaction} is a pair $(\O,\po)$ where $\O \subseteq \Op$ is a finite non-empty set of operations and $\po \subseteq \O\times\O$ is a strict total order called the \emph{program order}.


For $T=(\O,\po)$ and  $o\in\O$ on key $x$, let $\po^{-1}_{x}(o) \triangleq \{o'\in\O \mid o' = \_(x,\_) \land o' \rel{\po} o\}$ be the set of operations on $x$ preceding $o$ in $\po$.
We write $T \vdash \W(x,v)$ if $T$ writes $v$ to $x$, and $T \vdash \R(x,v)$ if $T$ reads value $v$ from $x$. 
Let $\WriteTx_{x} \triangleq \{T \mid T \vdash \W(x,\_)\}$. Transactions $T$ and $T'$ \emph{write-write conflict}, written $T \wwconflict T'$, if $T,T'\in\WriteTx_x$ for some $x$.

\inlsec{Histories}
Clients interact with the KVS by issuing transactions in sessions. 
A \emph{history} records the client-visible outcomes of these interactions.
Formally,	a \emph{history} is a triple $\H=(\T,\SO,\level)$ where $\T$ is a set of transactions with disjoint operation sets, $\SO\subseteq\T\times\T$ is the session order (a union of strict total orders, one per session), and $\level:\T\to\IL$ assigns an isolation level to each transaction. 

We consider $\IL\triangleq\allLevelsSet$, ranged over by $\il$. We often write $T[\il]$ to denote a transaction $T$ with $\level(T)=\il$ and call it an $\il$-transaction.
\vspace{-1ex}

\section{\ourframework: Mixing Isolation Levels }
\label{sec-mil}
This section presents the \ourframework semantic framework for mixed isolation levels.
 At a high level, 
it extends the $(\VIS, \AR)$ axiomatic framework  in~\cite{Framework:CONCUR2015}
with 
\emph{per-transaction} visibility. 


\inlsec{Abstract Executions}
An  execution is modeled using two relations over transactions:
\emph{visibility} ($\VIS$), which specifies whose effects are observed by each transaction,
and \emph{arbitration} ($\AR$), which  defines  a global commit order.
Formally,
an  \bfit{abstract execution} is $\AE = (\T, \SO, \level, \allowbreak \VIS, \AR)$, 
where $(\T, \SO, \level)$ is a history,
$\VIS \subseteq \T \times \T$ is an acyclic relation,
and $\AR \subseteq \T \times \T$ is a strict total order with $\VIS \subseteq \AR$.

We write $T' \rel{\VIS} T$ for $(T', T) \in \VIS$, and similarly for $\AR$.
$\VIS^{-1}(T) \triangleq \set{S \mid S \rel{\VIS} T}$ is the \emph{visible set} of $T$.
$T' \rel{\nVIS} T$ abbreviates $T' \rel{\AR} T \land \lnot(T' \rel{\VIS} T)$.





\small
\begin{figure}[t]
	\centering
	\begin{minipage}{1.0\columnwidth}
		\begin{align*}
			\Int(T) \equiv\;         & \forall r, x, v.\;
			(r = \R(x, v) \land \po^{-1}_{x}(r) \neq \emptyset) \implies                                                    \\
			                         & \; \max_{\po}(\po^{-1}_{x}(r)) = \_(x,v).                                               \\
			\Ext(T) \equiv\;         & \forall x, v.\; T \vdash \R(x, v) \implies                                           \\
			                         & \; \max_{\AR}(\VIS^{-1}(T) \cap \WriteTx_{x}) \vdash \W(x, v).                          \\
			\Session(T) \equiv\;     & \forall T' \in \T.\;
			T' \rel{\SO} T
			\implies T' \rel{\VIS} T.                                                                                       \\
			\TransVis(T) \equiv\;    & \forall T', S \in \T.\;
			T' \rel{\VIS} S \rel{\VIS} T \implies T' \rel{\VIS} T.                                                          \\
			\Prefixaxiom(T) \equiv\; & \forall T', S \in \T.\;
			T' \rel{\AR} S \rel{\VIS} T \implies T' \rel{\VIS} T.                                                           \\
			\NoConflict(T) \equiv\;  & \forall T' \in \T.\; T' \wwconflict T \land T' \rel{\AR} T \implies T' \rel{\VIS} T. \\
			\TotalVis(T) \equiv\;    & \forall T' \in \T.\; T' \rel{\AR} T \implies T' \rel{\VIS} T.
		\end{align*}
	\end{minipage}
	\captionsetup{skip=5pt}
	\caption{Consistency axioms for individual transactions.}
	\label{fig:axioms-mil}
	\vspace{-5pt}
\end{figure}
\normalsize

\inlsec{Consistency Axioms}
An isolation level $\il$ for a transaction $T$ 
is specified as a set of axioms
constraining its visible set 
 $\VIS^{-1}(T)$. 
 The \bfit{consistency axioms} on $T$ are shown in Figure~\ref{fig:axioms-mil}.

A read $r = \R(k, \_)$ of $T$ is \emph{external} if $\po_{x}^{-1}(r) = \emptyset$
($T$ has not previously written to $k$), and \emph{internal} otherwise.
$\Int(T)$ requires $T$ to read its own previous writes (internal consistency).
$\Ext(T)$ requires each external read observe the $\AR$-maximal visible write to the same key.
$\Session(T)$ 
requires $T$ to observe all transactions that precede it in session order. 
$\TransVis(T)$ makes visibility transitive: if $T$ observes $S$, 
it
must also observe every transaction observed by $S$.
$\Prefixaxiom(T)$ forces the visible set of  $T$  to be a prefix of $\AR$.
$\NoConflict(T)$ forbids write-write conflicts between transactions concurrent with $T$.
$\TotalVis(T)$ requires $T$ to observe \emph{all} of its $\AR$-predecessors.


\begin{table*}[ht]
	\captionsetup{skip=3pt}
	\caption{Overview of the construction and verification steps for the two case studies. 
		$T.\sts$  and $T.\cts$ 
		denote the start and commit   timestamps of transaction $T$, respectively.
	}
	\label{tab:proof-comparison}
	\footnotesize
	\centering
	\renewcommand{\arraystretch}{1.25}
	\setlength{\tabcolsep}{5pt}
	\begin{tabularx}{\textwidth}{@{}p{0.17\textwidth}>{\hsize=0.9\hsize}X>{\hsize=1.1\hsize}X@{}}
		\toprule
		\textbf{Proof Step}
		& \textbf{SI--S2PL (Section~\ref{ss-si-ser-fekete})}
		& \textbf{PC--SI--SER (Appendix~\ref{app-pc-si-ser})} \\
		\midrule
		Constructing $\AR$
		& Commit timestamp order: $T'.\cts < T.\cts$.
		& Commit timestamp order: $T'.\cts < T.\cts$. \\
		Constructing $\VIS$
		&
		$T[\SI]$ reads from snapshot $T.\sts > T'.\cts$;
		$T[\SER]$ reads from the latest $T'.\cts < T.\cts$.
		& $T[\PC]$ and $T[\SI]$ read from snapshot $T.\sts > T'.\cts$;
		$T[\SER]$ effectively reads from the latest $T'.\cts < T.\cts$ due to conflict detection. \\   \midrule
		Verifying weaker axioms
		& $\Int$, $\Ext$, $\Session$, $\Prefixaxiom$, $\NoConflict$:
		$\NoConflict$ enforced by the \emph{first-committer-wins} rule for $\SI$.
		& $\Int$, $\Ext$, $\Session$, $\Prefixaxiom$, $\NoConflict$:
		$\Prefixaxiom$ required by both $\PC$ and $\SI$;
		$\NoConflict$ required only by $\SI$. \\
		Verifying axiom  $\TotalVis$ 
		& $\TotalVis$ via S2PL:
		$\SER$ holds exclusive locks until commit,
		ensuring visibility of all prior commits.
		& $\TotalVis$ via timestamp ordering:
		$\SER$ reads the latest value of every key,
		yielding $\forall T' \!\rel{\AR}\! T \!\implies\! T' \!\rel{\VIS}\! T$. \\
		\bottomrule
	\end{tabularx}
\end{table*}

\small
\begin{figure}[tb]
    \centering
    \captionsetup{skip=0pt}
    \[
    \begin{aligned}
        \RA(T) &\equiv \Int(T) \land \Ext(T) \land \Session(T)\\[4pt]
        \CC(T) &\equiv \RA(T) \land \TransVis(T)\\[4pt]
        \PC(T) &\equiv \RA(T) \land \Prefixaxiom(T)\\[4pt]
        \PSI(T) &\equiv \RA(T) \land \TransVis(T) \land \NoConflict(T)\\[4pt]
        \SI(T) &\equiv \RA(T) \land \Prefixaxiom(T) \land \NoConflict(T)\\[4pt]
        \SER(T) &\equiv \RA(T) \land \TotalVis(T)
    \end{aligned}
    \]
    \caption{Isolation levels for individual transactions.}
    \label{fig:isolevel}
\end{figure}
\normalsize

\inlsec{Formalizing Mixed Isolation Levels}
We now define consistency for abstract executions and histories. 
An abstract execution $\AE = (\T, \SO, \level, \allowbreak \VIS, \AR)$
is  \bfit{consistent} if, 
 for each transaction $T \in \T$,
all  consistency axioms
 associated with 
its isolation level $\level(T)$ hold on $T$. 
Figure~\ref{fig:isolevel} presents
 the consistency axioms for each isolation level. 
A history $\H$ is then considered
 \bfit{consistent} if there exists a consistent abstract execution 
$\AE = (\H, \VIS, \AR)$.

\begin{definition}
\textit{A concurrency control protocol conforms to a set of  mixed isolation guarantees 
if	 every history it produces is consistent 
with respect to them. }
\end{definition}

In \ourframework, 
$\VIS$ determines visibility on a per-transaction basis, 
whereas $\AR$ provides a shared total order over all transactions.
Homogeneous isolation settings correspond to the special case where all transactions are assigned the same isolation level.
Thus, 
\ourframework generalizes the framework of~\cite{Framework:CONCUR2015} 
(see Appendix~\ref{app-mil-il}).
To validate \ourframework, 
we also establish its equivalence to a recent formalization of mixed isolation levels~\cite{ComplexityMIL:CAV2025} on their common isolation levels (i.e., $\RA$, $\PC$, $\SI$, and $\SER$);
the proof is given in Appendix~\ref{app-mil-cav-equiv}.



\begin{figure}[t]
	\vspace{-2ex}
	\centering
	\includegraphics[width=.75\columnwidth]{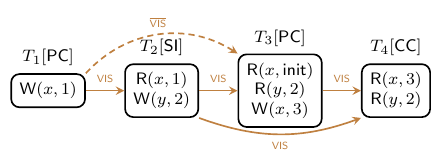}
	\captionsetup{skip=0pt}
	\caption{%
		 An inconsistent execution: $\Prefixaxiom(T_3)$ is violated.
	}
	\label{fig:consistent}
	    	\vspace{-2ex}
\end{figure}

 \paragraph{Example}
Figures~\ref{fig:intro-example} and \ref{fig:consistent} show two abstract executions.
Both share $T_1[\PC] \rel{\VIS} T_2[\SI] \rel{\VIS} T_4[\CC]$
and differ only in $\level(T_3)$.
The execution in Figure~\ref{fig:intro-example} is consistent 
as all axioms required by each transaction's isolation level are satisfied.
In contrast, 
 the execution in Figure~\ref{fig:consistent} 
is inconsistent.
$T_3[\PC]$  observes $T_2$ (via $\R(y,2)$) but not $T_1$,
which violates
 $\Prefixaxiom(T_3)$ as
$T_1 \rel{\AR} T_2 \rel{\VIS} T_3$
implies
$T_1 \rel{\VIS} T_3$.

\section{Checking Semantic Conformance}
\label{sec-conformance}

This section demonstrates how \ourframework can be used to establish the semantic conformance of concurrency control protocols to their intended mixed isolation guarantees.


\subsection{Proof Approach }
\label{ss-proof-method}

Our semantic conformance proof follows prior construction-based approaches~\cite{burckhardtPrinEvenCons2014,xiongECOOP2020},
which
consists of two main steps.

\inlsec{\ding{182} Construction}
We extract the $\VIS$ and $\AR$ relations from the protocol specification, 
 such as its pseudocode or operational semantics.
Typically,
  $\AR$ is obtained from the commit order of transactions 
  (e.g., commit timestamps),
   whereas $\VIS$ is derived by a case analysis on $\level(T)$ to
   capture the visibility requirements of the corresponding isolation level.
   
   After constructing $\VIS$ and $\AR$,
    one must first show that they form a valid abstract execution:
    (i) $\VIS$ is acyclic, 
(ii)    $\AR$ is a strict total order,
     and 
(iii)     $\VIS \subseteq \AR$.
   For example, (i) follows from deriving $\VIS$ from a strict partial order, 
   (ii) from unique commit timestamps, 
   and
     (iii) from the fact that visibility is bounded by commit order.
 
%

  \inlsec{\ding{183} Verification}
We show that the extracted relations  satisfy the consistency axioms
  required by 
  each transaction's
  assigned isolation level. 
Typically, 
 $\Int(T)$ follows from local write buffering,
  while $\Session(T)$, $\Prefixaxiom(T)$, and $\TransVis(T)$
    reduce to  reasoning over the construction of $\VIS$.
For instance, when $\VIS$ is derived from timestamps,
  these axioms follow from timestamp monotonicity and transitivity.
More generally, any construction that derives $\VIS$ from a transitive order, e.g., logical clocks, establishes them directly.
The main proof obligations are therefore $\Ext(T)$, 
 showing that each external read observes the $\AR$-maximal visible write,
   and
     the level-specific axioms (e.g., $\NoConflict(T)$ and $\TotalVis(T)$), 
      which are typically established using protocol-specific mechanisms 
      for conflict detection and resolution.

Table~\ref{tab:proof-comparison} outlines these  steps for  two case studies:
the SI--S2PL protocol~\cite{Allocating:PODS2005,ModOrMixing:DASFAA2008}, which mixes \SI{} and \SER{} and is used in Microsoft SQL Server and Oracle Berkeley DB, and a new protocol that combines  \PC, \SI, and \SER. 
In what follows,
  we focus on the former and defer the latter to Appendix~\ref{app-pc-si-ser}.

\subsection{Case Study: The SI--S2PL Protocol }
\label{ss-si-ser-fekete}


Algorithm~\ref{alg:si-ser-fekete} shows the pseudocode of the SI--S2PL protocol.
Its key idea is
to unify $\SI$ (optimistic, snapshot-based) and $\SER$ (pessimistic, lock-based)
through a  versioned store
$\store : \Key \to (\TS \partialmapsto \Val)$.
Transactions of both levels buffer writes locally in $T.\buffer$
(line~\ref{line:writetxn-buffer-si-ser-fekete}),
obtain commit timestamps ($T.\cts$)
from a shared clock
(line~\ref{line:committxn-committime-si-ser-fekete}),
and
install committed versions into  the store at their commit timestamps
(line~\ref{line:committxn-si-install-writes-si-ser-fekete}).
An $\SI$ transaction additionally records a start timestamp $T.\sts$ that defines its snapshot
(line~\ref{line:starttxn-sts-si-ser-fekete}).

During execution,
writes are buffered in $T.\buffer$
(line~\ref{line:writetxn-buffer-si-ser-fekete}).
An $\SER$ transaction acquires $\slock$  (shared locks) on reads
(line~\ref{line:readtxn-ser-slock})
and $\xlock$ (exclusive locks) on writes
(line~\ref{line:writetxn-ser-xlock}),
whereas an $\SI$ transaction performs reads without locking and acquires $\xlock$ on its write set only during commit
(line~\ref{line:committxn-si-xlock}).
All locks obey standard compatibility rules
and are held until commit (S2PL).
For reads, buffered values are returned whenever available
(line~\ref{line:readtxn-buffer-si-ser-fekete}).
Otherwise, an $\SI$ transaction reads the latest version whose timestamp precedes $T.\sts$
(line~\ref{line:readtxn-si-store-si-ser-fekete}),
while
an $\SER$ transaction reads
the latest version after acquiring the corresponding shared lock
(lines~\ref{line:readtxn-ser-slock}--\ref{line:readtxn-ser-store-si-ser-fekete}).

At commit,
an $\SI$ transaction first acquires $\xlock$ on its write set
(line~\ref{line:committxn-si-xlock}).
Both levels then obtain a fresh commit timestamp 
(line~\ref{line:committxn-committime-si-ser-fekete}).
An $\SI$ transaction then performs the
\emph{first-committer-wins}
check~\cite{Critique:SIGMOD1995}
(line~\ref{line:committxn-write-conflict-checking-si-ser-fekete}):
if any concurrent transaction $T'$
whose commit timestamp
lies in
$(T.\sts, T.\cts)$
wrote  to a key
in  $T$'s   write set,
then $T$ aborts.
An $\SER$ transaction skips this check, as conflicting writes are already prevented by exclusive locks.
Finally, the buffered writes are installed into the store
and all locks are released
(lines~\ref{line:committxn-si-install-writes-si-ser-fekete}--\ref{line:committxn-releaselocks}).

\begin{algorithm}[!t]
  \footnotesize\raggedright
  \caption{The SI-S2PL protocol \cite{Allocating:PODS2005,ModOrMixing:DASFAA2008} }
  \label{alg:si-ser-fekete}
  \begin{algorithmic}[1]
		\Statex $\store: \Key \to (\TS \partialmapsto \Val)$:
		  versioned KV store, mapping each key to timestamped values
		\Statex $T.\buffer: \Key \partialmapsto \Val$:
		  write buffer of transaction $T$ 

		\hStatex
    \Procedure{\StartProc}{$T$}
      \label{line:starttxn-si-ser-fekete}
			\If{$\level(T) = \SI$}
			  \Comment{no $\sts$ for \SER{} transactions}
				\State $T.\sts \gets \now$
					\label{line:starttxn-sts-si-ser-fekete}
			\EndIf
    \EndProcedure

		\hStatex
    \Procedure{\WriteProc}{$T, k, v$}
      \label{line:writetxn-si-ser-fekete}
			\If{$\level(T) = \SER$}
			  \label{line:writetxn-ser}
				\State $T.\xlock(k)$
					\Comment{acquire an exclusive lock on $k$}
					\label{line:writetxn-ser-xlock}
			\EndIf
			\State $T.\buffer[k] \gets v$
				\Comment{buffer $v$ and overwrite previous writes to $k$}
				\label{line:writetxn-buffer-si-ser-fekete}
    \EndProcedure

		\hStatex
    \Procedure{\ReadProc}{$T, k$}
		  \label{line:proc-readtxn-si-ser-fekete}
			\If{$k \in \dom(T.\buffer)$}
			  \Comment{internal read from my own write}
				\State \Return $T.\buffer[k]$
				\label{line:readtxn-buffer-si-ser-fekete}
			\EndIf
			\If{$\level(T) = \SI$}
			  \Comment{external read from snapshot}
				\State \Return value of $\store[k]$ at latest timestamp $< T.\sts$
					\label{line:readtxn-si-store-si-ser-fekete}
			\EndIf
			\If{$\level(T) = \SER$}
			  \Comment{external read from latest committed value}
				\State $T.\slock(k)$
					\Comment{acquire a shared lock on $k$}
				  \label{line:readtxn-ser-slock}
				\State \Return value of $\store[k]$ at latest timestamp
				  \label{line:readtxn-ser-store-si-ser-fekete}
			\EndIf
    \EndProcedure

		\hStatex
    \Procedure{\CommitProc}{$T$}
      \label{line:proc-committxn-si-ser-fekete}
			\If{$\level(T) = \SI$}
			  \label{line:committxn-si-buffer}
				\State $T.\xlock(k), \forall k \in \dom(T.\buffer)$
				\Comment{acquire exclusive locks}
				\label{line:committxn-si-xlock}
			\EndIf
			\State $T.\cts \gets \now$
				\Comment{$\level(T) \in \set{\SI, \SER}$}
				\label{line:committxn-committime-si-ser-fekete}
			\If{$\level(T) = \SI$}
				\Comment{$\level(T') \in \set{\SI, \SER}$}
				\If{$\exists T' \wwconflict T.\; T'.\cts \in (T.\sts, T.\cts)$}
					\label{line:committxn-write-conflict-checking-si-ser-fekete}
					\State $T.\releaseLocks()$
					\State \Return $\aborted$
			\EndIf
		\EndIf
		\State $\store[k] \gets [T.\cts \mapsto v]$ for all $[k \mapsto v] \in T.\buffer$
			\Comment{install writes} 
			\label{line:committxn-si-install-writes-si-ser-fekete}
		\State $T.\releaseLocks()$
			\Comment{$\level(T) \in \set{\SI, \SER}$}
			\label{line:committxn-releaselocks}
		\State \Return $\committed$
    \EndProcedure
  \end{algorithmic}
\end{algorithm}
\normalsize



\begin{theorem}
	\label{thm:si-ser-fekete}
Algorithm~\ref{alg:si-ser-fekete} conforms to the mixed isolation guarantees of \SI{} and \SER.
\end{theorem}

\smallskip
\noindent
\textit{Proof Sketch.}
We follow  the proof steps  shown in Table~\ref{tab:proof-comparison}.

\smallskip\noindent
\textbf{$\bullet$ Constructing $\AR$.}
$\AR$  follows  commit-timestamp order:
$T' \rel{\AR} T \iff T'.\cts < T.\cts$.
Since commit timestamps are unique and  
strictly increasing, 
$\AR$ is a strict total order.

\smallskip\noindent
\textbf{$\bullet$ Constructing $\VIS$.}
$\VIS$ is obtained by a case analysis on $\level(T)$:
{\abovedisplayskip=2pt \belowdisplayskip=2pt
\begin{align*}
  T' \rel{\VIS} T \iff
   & \; (\level(T) = \SI \land T'.\cts < T.\sts) \;\lor\; \\
   & \; (\level(T) = \SER \land T'.\cts < T.\cts).
\end{align*}}

For $\SI$, $\VIS$ captures the snapshot at  $T.\sts$
(line~\ref{line:readtxn-si-store-si-ser-fekete}).
For $\SER$, $\VIS$ includes all prior commits
(line~\ref{line:readtxn-ser-store-si-ser-fekete})
using
 $T'.\cts < T.\cts$  
 rather than $T'.\cts < T.\sts$
so that 
 every $\AR$-predecessor of $T$ is visible.
 By construction, 
$\VIS$ is irreflexive
and $\VIS \subseteq \AR$.

\smallskip\noindent
\textbf{$\bullet$  Verifying Weaker Axioms.}
\begin{itemize}[label=--]
  \item $\Int(T)$ follows from the local write buffer
        (line~\ref{line:readtxn-buffer-si-ser-fekete}).

  \item $\Ext(T)$:
  For $T[\SI]$, the snapshot read
                (line~\ref{line:readtxn-si-store-si-ser-fekete})
returns the version with the largest timestamp below $T.\sts$,
i.e., $\max_{\AR}(\VIS^{-1}(T) \cap \WriteTx_x)$.
For  $T[\SER]$,
$T$ acquires $\slock$ before reading
 the latest value
                (line~\ref{line:readtxn-ser-slock})
                 and holds it until commit. 
                 Hence, no transaction can install a later version of $x$ before $T$ commits,
                 so the read returns the $\AR$-maximal visible writer.
                 

  \item $\Session(T)$: $T' \rel{\SO} T$ implies  that $T$ starts after $T'$ commits.
        By timestamp monotonicity, $T'.\cts < T.\sts$ (if $\SI$) or $T'.\cts < T.\cts$ (if $\SER$),
        so $T' \rel{\VIS} T$.

  \item $\Prefixaxiom(T)$ for $T[\SI]$:
        $T' \rel{\AR} S \rel{\VIS} T$ implies
        $T'.\cts < S.\cts$ and $S.\cts < T.\sts$. 
        Hence, $T'.\cts < T.\sts$ and $T' \rel{\VIS} T$.

  \item $\NoConflict(T)$ for $T[\SI]$:
        Let $T' \wwconflict T$ with $T' \rel{\AR} T$.
        If $T.\sts < T'.\cts$, then $T'.\cts \in (T.\sts, T.\cts)$ and the
        first-committer-wins check
        (line~\ref{line:committxn-write-conflict-checking-si-ser-fekete})
        forces $T$ to abort.
        Hence,  $T'.\cts < T.\sts$, yielding $T' \rel{\VIS} T$.
\end{itemize}

\smallskip\noindent
\textbf{$\bullet$ Verifying Axiom $\TotalVis(T)$.}
For $\level(T) = \SER$, the two construction steps  yield  $T' \rel{\AR} T \iff T'.\cts < T.\cts$
and $T'.\cts < T.\cts \iff T' \rel{\VIS} T$, respectively. 
Hence, $\VIS^{-1}(T) = \AR^{-1}(T)$, satisfying $\TotalVis(T)$.

\smallskip
The full proof is given in Appendix~\ref{app-proof-si-ser-fekete}.


\section{Future Work }
\label{sec:concl}

%
%

 We have presented  \ourframework
 for reasoning about mixed isolation levels
 and  
demonstrated
   its applicability through two case studies.  
  \ourframework is designed as a \emph{general} 
semantic framework.
It can also be used to verify the semantic conformance of concurrency control protocols under   homogeneous
isolation settings.
 Our proofs currently rely on manual reasoning;
 automating the verification process 
is a natural next step.

In addition, 
   \ourframework 
   can support 
a variety of other
 applications.
These include
(i)
 \emph{black-box testing}
 of mixed isolation guarantees~\cite{ComplexityMIL:CAV2025}
(e.g., by deriving a dependency-graph-based characterization  in the style of~\cite{AnalysingSI:JACM2018} that is equivalent to \ourframework),
(ii) \emph{robustness checking}~\cite{Allocating:PODS2005} for safely  allocating a wide range of different isolation levels across transaction workloads
 (e.g., by following the approach in~\cite{RobustRA:CONCUR2016}),
 and
 (iii)
  \emph{adaptive concurrency control}~\cite{TxnSails:VLDB2025}
  (e.g., by extending robustness checking to online settings).

\bibliographystyle{ACM-Reference-Format}
\bibliography{ref}


\begin{thebibliography}{33}


\ifx \showCODEN    \undefined \def \showCODEN     #1{\unskip}     \fi
\ifx \showDOI      \undefined \def \showDOI       #1{#1}\fi
\ifx \showISBNx    \undefined \def \showISBNx     #1{\unskip}     \fi
\ifx \showISBNxiii \undefined \def \showISBNxiii  #1{\unskip}     \fi
\ifx \showISSN     \undefined \def \showISSN      #1{\unskip}     \fi
\ifx \showLCCN     \undefined \def \showLCCN      #1{\unskip}     \fi
\ifx \shownote     \undefined \def \shownote      #1{#1}          \fi
\ifx \showarticletitle \undefined \def \showarticletitle #1{#1}   \fi
\ifx \showURL      \undefined \def \showURL       {\relax}        \fi
\providecommand\bibfield[2]{#2}
\providecommand\bibinfo[2]{#2}
\providecommand\natexlab[1]{#1}
\providecommand\showeprint[2][]{arXiv:#2}

\bibitem[\protect\citeauthoryear{Adya}{Adya}{1999}]%
        {Adya:PhDThesis1999}
\bibfield{author}{\bibinfo{person}{A. Adya}.} \bibinfo{year}{1999}\natexlab{}.
\newblock \emph{\bibinfo{title}{Weak Consistency: A Generalized Theory and Optimistic Implementations for Distributed Transactions}}.
\newblock \bibinfo{thesistype}{Ph.D. Dissertation}. \bibinfo{school}{Massachusetts Institute of Technology}, \bibinfo{address}{USA}.
\newblock


\bibitem[\protect\citeauthoryear{Akkoorath, Tomsic, Bravo, Li, Crain, Bieniusa, Pregui{\c{c}}a, and Shapiro}{Akkoorath et~al\mbox{.}}{2016}]%
        {Cure:ICDCS2016}
\bibfield{author}{\bibinfo{person}{Deepthi~Devaki Akkoorath}, \bibinfo{person}{Alejandro~Z. Tomsic}, \bibinfo{person}{Manuel Bravo}, \bibinfo{person}{Zhongmiao Li}, \bibinfo{person}{Tyler Crain}, \bibinfo{person}{Annette Bieniusa}, \bibinfo{person}{Nuno~M. Pregui{\c{c}}a}, {and} \bibinfo{person}{Marc Shapiro}.} \bibinfo{year}{2016}\natexlab{}.
\newblock \showarticletitle{Cure: Strong Semantics Meets High Availability and Low Latency}. In \bibinfo{booktitle}{\emph{{ICDCS} 2016}}. \bibinfo{publisher}{{IEEE} Computer Society}, \bibinfo{pages}{405--414}.
\newblock
\urldef\tempurl%
\url{https://doi.org/10.1109/ICDCS.2016.98}
\showDOI{\tempurl}


\bibitem[\protect\citeauthoryear{Alnor~Mathiasen, Gondelman, Ducruet, Timany, and Birkedal}{Alnor~Mathiasen et~al\mbox{.}}{2025}]%
        {10.1145/3747515}
\bibfield{author}{\bibinfo{person}{Anders Alnor~Mathiasen}, \bibinfo{person}{L\'{e}on Gondelman}, \bibinfo{person}{L\'{e}on Ducruet}, \bibinfo{person}{Amin Timany}, {and} \bibinfo{person}{Lars Birkedal}.} \bibinfo{year}{2025}\natexlab{}.
\newblock \showarticletitle{Reasoning about Weak Isolation Levels in Separation Logic}.
\newblock \bibinfo{journal}{\emph{Proc. ACM Program. Lang.}} \bibinfo{volume}{9}, \bibinfo{number}{ICFP}, Article \bibinfo{articleno}{246} (\bibinfo{date}{Aug.} \bibinfo{year}{2025}), \bibinfo{numpages}{35}~pages.
\newblock
\urldef\tempurl%
\url{https://doi.org/10.1145/3747515}
\showDOI{\tempurl}


\bibitem[\protect\citeauthoryear{Alomari, Cahill, Fekete, and R{\"o}hm}{Alomari et~al\mbox{.}}{2008}]%
        {ModOrMixing:DASFAA2008}
\bibfield{author}{\bibinfo{person}{Mohammad Alomari}, \bibinfo{person}{Michael Cahill}, \bibinfo{person}{Alan Fekete}, {and} \bibinfo{person}{Uwe R{\"o}hm}.} \bibinfo{year}{2008}\natexlab{}.
\newblock \showarticletitle{Serializable {{Executions}} with {{Snapshot Isolation}}: {{Modifying Application Code}} or {{Mixing Isolation Levels}}?}. In \bibinfo{booktitle}{\emph{DASFAA 2008}}, Vol.~\bibinfo{volume}{4947}. \bibinfo{pages}{267--281}.
\newblock
\urldef\tempurl%
\url{https://doi.org/10.1007/978-3-540-78568-2_21}
\showDOI{\tempurl}


\bibitem[\protect\citeauthoryear{Bailis, Fekete, Ghodsi, Hellerstein, and Stoica}{Bailis et~al\mbox{.}}{2016}]%
        {RAMP:TODS2016}
\bibfield{author}{\bibinfo{person}{Peter Bailis}, \bibinfo{person}{Alan Fekete}, \bibinfo{person}{Ali Ghodsi}, \bibinfo{person}{Joseph~M. Hellerstein}, {and} \bibinfo{person}{Ion Stoica}.} \bibinfo{year}{2016}\natexlab{}.
\newblock \showarticletitle{Scalable {{Atomic Visibility}} with {{RAMP Transactions}}}.
\newblock \bibinfo{journal}{\emph{ACM Trans. Database Syst.}} \bibinfo{volume}{41}, \bibinfo{number}{3} (\bibinfo{date}{July} \bibinfo{year}{2016}), \bibinfo{pages}{15:1--15:45}.
\newblock
\urldef\tempurl%
\url{https://doi.org/10.1145/2909870}
\showDOI{\tempurl}


\bibitem[\protect\citeauthoryear{Berenson, Bernstein, Gray, Melton, O'Neil, and O'Neil}{Berenson et~al\mbox{.}}{1995}]%
        {Critique:SIGMOD1995}
\bibfield{author}{\bibinfo{person}{Hal Berenson}, \bibinfo{person}{Phil Bernstein}, \bibinfo{person}{Jim Gray}, \bibinfo{person}{Jim Melton}, \bibinfo{person}{Elizabeth O'Neil}, {and} \bibinfo{person}{Patrick O'Neil}.} \bibinfo{year}{1995}\natexlab{}.
\newblock \showarticletitle{A critique of ANSI SQL isolation levels}.
\newblock \bibinfo{journal}{\emph{SIGMOD Rec.}} \bibinfo{volume}{24}, \bibinfo{number}{2} (\bibinfo{date}{May} \bibinfo{year}{1995}), \bibinfo{pages}{1–10}.
\newblock
\urldef\tempurl%
\url{https://doi.org/10.1145/568271.223785}
\showDOI{\tempurl}


\bibitem[\protect\citeauthoryear{Bernardi and Gotsman}{Bernardi and Gotsman}{2016}]%
        {RobustRA:CONCUR2016}
\bibfield{author}{\bibinfo{person}{Giovanni Bernardi} {and} \bibinfo{person}{Alexey Gotsman}.} \bibinfo{year}{2016}\natexlab{}.
\newblock \showarticletitle{{Robustness against Consistency Models with Atomic Visibility}}. In \bibinfo{booktitle}{\emph{CONCUR 2016}}, Vol.~\bibinfo{volume}{59}. \bibinfo{pages}{7:1--7:15}.
\newblock
\urldef\tempurl%
\url{https://doi.org/10.4230/LIPIcs.CONCUR.2016.7}
\showDOI{\tempurl}


\bibitem[\protect\citeauthoryear{Bouajjani, Enea, and Román-Calvo}{Bouajjani et~al\mbox{.}}{2025}]%
        {ComplexityMIL:CAV2025}
\bibfield{author}{\bibinfo{person}{Ahmed Bouajjani}, \bibinfo{person}{Constantin Enea}, {and} \bibinfo{person}{Enrique Román-Calvo}.} \bibinfo{year}{2025}\natexlab{}.
\newblock \showarticletitle{On the {{Complexity}} of {{Checking Mixed Isolation Levels}} for {{SQL Transactions}}}. In \bibinfo{booktitle}{\emph{CAV 2025}}, Vol.~\bibinfo{volume}{15934}. \bibinfo{pages}{315--337}.
\newblock
\urldef\tempurl%
\url{https://doi.org/10.1007/978-3-031-98685-7_15}
\showDOI{\tempurl}


\bibitem[\protect\citeauthoryear{Burckhardt}{Burckhardt}{2014}]%
        {burckhardtPrinEvenCons2014}
\bibfield{author}{\bibinfo{person}{Sebastian Burckhardt}.} \bibinfo{year}{2014}\natexlab{}.
\newblock \showarticletitle{Principles of {Eventual Consistency}}.
\newblock \bibinfo{journal}{\emph{Foundations and Trends in Programming Languages}} \bibinfo{volume}{1}, \bibinfo{number}{1-2} (\bibinfo{date}{Oct.} \bibinfo{year}{2014}), \bibinfo{pages}{1--150}.
\newblock
\showISSN{2325-1107}
\urldef\tempurl%
\url{https://doi.org/10.1561/2500000011}
\showDOI{\tempurl}


\bibitem[\protect\citeauthoryear{Burckhardt, Leijen, Protzenko, and F{\"a}hndrich}{Burckhardt et~al\mbox{.}}{2015}]%
        {GSP:ECOOP2015}
\bibfield{author}{\bibinfo{person}{Sebastian Burckhardt}, \bibinfo{person}{Daan Leijen}, \bibinfo{person}{Jonathan Protzenko}, {and} \bibinfo{person}{Manuel F{\"a}hndrich}.} \bibinfo{year}{2015}\natexlab{}.
\newblock \showarticletitle{Global {{Sequence Protocol}}: {{A Robust Abstraction}} for {{Replicated Shared State}}}. In \bibinfo{booktitle}{\emph{ECOOP 2015}}, Vol.~\bibinfo{volume}{37}. \bibinfo{pages}{568--590}.
\newblock
\urldef\tempurl%
\url{https://doi.org/10.4230/LIPIcs.ECOOP.2015.568}
\showDOI{\tempurl}


\bibitem[\protect\citeauthoryear{Center}{Center}{2025}]%
        {OracleTxnIsoLevel}
\bibfield{author}{\bibinfo{person}{Oracle~Help Center}.} \bibinfo{year}{Accessed in December 2025}\natexlab{}.
\newblock \bibinfo{title}{Data {{Concurrency}} and {{Consistency}}}.
\newblock \bibinfo{howpublished}{\url{https://docs.oracle.com/en/database/oracle/oracle-database/23/cncpt/data-concurrency-and-consistency.html}}.
\newblock


\bibitem[\protect\citeauthoryear{Cerone, Bernardi, and Gotsman}{Cerone et~al\mbox{.}}{2015}]%
        {Framework:CONCUR2015}
\bibfield{author}{\bibinfo{person}{Andrea Cerone}, \bibinfo{person}{Giovanni Bernardi}, {and} \bibinfo{person}{Alexey Gotsman}.} \bibinfo{year}{2015}\natexlab{}.
\newblock \showarticletitle{{A Framework for Transactional Consistency Models with Atomic Visibility}}. In \bibinfo{booktitle}{\emph{CONCUR 2015}}, Vol.~\bibinfo{volume}{42}. \bibinfo{pages}{58--71}.
\newblock
\urldef\tempurl%
\url{https://doi.org/10.4230/LIPIcs.CONCUR.2015.58}
\showDOI{\tempurl}


\bibitem[\protect\citeauthoryear{Cerone and Gotsman}{Cerone and Gotsman}{2018}]%
        {AnalysingSI:JACM2018}
\bibfield{author}{\bibinfo{person}{Andrea Cerone} {and} \bibinfo{person}{Alexey Gotsman}.} \bibinfo{year}{2018}\natexlab{}.
\newblock \showarticletitle{Analysing Snapshot Isolation}.
\newblock \bibinfo{journal}{\emph{J. ACM}} \bibinfo{volume}{65}, \bibinfo{number}{2}, Article \bibinfo{articleno}{11} (\bibinfo{date}{Jan} \bibinfo{year}{2018}), \bibinfo{numpages}{41}~pages.
\newblock
\urldef\tempurl%
\url{https://doi.org/10.1145/3152396}
\showDOI{\tempurl}


\bibitem[\protect\citeauthoryear{CockroachDB}{CockroachDB}{2025}]%
        {CockroachDBTxnIsoLevel}
\bibfield{author}{\bibinfo{person}{CockroachDB}.} \bibinfo{year}{Accessed in December, 2025}\natexlab{}.
\newblock \bibinfo{title}{{{CockroachDB Docs - SET TRANSACTION}}}.
\newblock \bibinfo{howpublished}{\url{https://www.cockroachlabs.com/docs/stable/set-transaction}}.
\newblock


\bibitem[\protect\citeauthoryear{Crooks, Pu, Alvisi, and Clement}{Crooks et~al\mbox{.}}{2017}]%
        {DBLP:conf/podc/CrooksPAC17}
\bibfield{author}{\bibinfo{person}{Natacha Crooks}, \bibinfo{person}{Youer Pu}, \bibinfo{person}{Lorenzo Alvisi}, {and} \bibinfo{person}{Allen Clement}.} \bibinfo{year}{2017}\natexlab{}.
\newblock \showarticletitle{Seeing is Believing: {A} Client-Centric Specification of Database Isolation}. In \bibinfo{booktitle}{\emph{{PODC}'17}}. \bibinfo{publisher}{{ACM}}, \bibinfo{pages}{73--82}.
\newblock
\urldef\tempurl%
\url{https://doi.org/10.1145/3087801.3087802}
\showDOI{\tempurl}


\bibitem[\protect\citeauthoryear{Fekete}{Fekete}{2005}]%
        {Allocating:PODS2005}
\bibfield{author}{\bibinfo{person}{Alan Fekete}.} \bibinfo{year}{2005}\natexlab{}.
\newblock \showarticletitle{Allocating isolation levels to transactions}. In \bibinfo{booktitle}{\emph{PODS 2005}}. \bibinfo{pages}{206–215}.
\newblock
\urldef\tempurl%
\url{https://doi.org/10.1145/1065167.1065193}
\showDOI{\tempurl}


\bibitem[\protect\citeauthoryear{Ghasemirad, Liu, Sprenger, Multazzu, and Basin}{Ghasemirad et~al\mbox{.}}{2025a}]%
        {VerIso:VLDB2025}
\bibfield{author}{\bibinfo{person}{Shabnam Ghasemirad}, \bibinfo{person}{Si Liu}, \bibinfo{person}{Christoph Sprenger}, \bibinfo{person}{Luca Multazzu}, {and} \bibinfo{person}{David Basin}.} \bibinfo{year}{2025}\natexlab{a}.
\newblock \showarticletitle{VerIso: Verifiable Isolation Guarantees for Database Transactions}.
\newblock \bibinfo{journal}{\emph{Proc. {VLDB} Endow.}} \bibinfo{volume}{18}, \bibinfo{number}{5} (\bibinfo{year}{2025}), \bibinfo{pages}{1362--1375}.
\newblock
\urldef\tempurl%
\url{https://doi.org/10.14778/3718057.3718065}
\showDOI{\tempurl}


\bibitem[\protect\citeauthoryear{Ghasemirad, Sprenger, Liu, Multazzu, and Basin}{Ghasemirad et~al\mbox{.}}{2025b}]%
        {10.1007/978-3-031-90660-2_3}
\bibfield{author}{\bibinfo{person}{Shabnam Ghasemirad}, \bibinfo{person}{Christoph Sprenger}, \bibinfo{person}{Si Liu}, \bibinfo{person}{Luca Multazzu}, {and} \bibinfo{person}{David Basin}.} \bibinfo{year}{2025}\natexlab{b}.
\newblock \showarticletitle{Pushing the Limit: Verified Performance-Optimal Causally-Consistent Database Transactions}. In \bibinfo{booktitle}{\emph{TACAS 2025}}. \bibinfo{publisher}{Springer-Verlag}, \bibinfo{address}{Berlin, Heidelberg}, \bibinfo{pages}{43–62}.
\newblock
\showISBNx{978-3-031-90659-6}
\urldef\tempurl%
\url{https://doi.org/10.1007/978-3-031-90660-2_3}
\showDOI{\tempurl}


\bibitem[\protect\citeauthoryear{Liu}{Liu}{2022}]%
        {10.1145/3494517}
\bibfield{author}{\bibinfo{person}{Si Liu}.} \bibinfo{year}{2022}\natexlab{}.
\newblock \showarticletitle{All in One: Design, Verification, and Implementation of SNOW-optimal Read Atomic Transactions}.
\newblock \bibinfo{journal}{\emph{ACM Trans. Softw. Eng. Methodol.}} \bibinfo{volume}{31}, \bibinfo{number}{3}, Article \bibinfo{articleno}{43} (\bibinfo{date}{March} \bibinfo{year}{2022}), \bibinfo{numpages}{44}~pages.
\newblock
\showISSN{1049-331X}
\urldef\tempurl%
\url{https://doi.org/10.1145/3494517}
\showDOI{\tempurl}


\bibitem[\protect\citeauthoryear{Liu, Multazzu, Wei, and Basin}{Liu et~al\mbox{.}}{2024}]%
        {noc-noc}
\bibfield{author}{\bibinfo{person}{Si Liu}, \bibinfo{person}{Luca Multazzu}, \bibinfo{person}{Hengfeng Wei}, {and} \bibinfo{person}{David~A. Basin}.} \bibinfo{year}{2024}\natexlab{}.
\newblock \showarticletitle{NOC-NOC: Towards Performance-optimal Distributed Transactions}.
\newblock \bibinfo{journal}{\emph{Proc. ACM Manag. Data}} \bibinfo{volume}{2}, \bibinfo{number}{1}, Article \bibinfo{articleno}{9} (\bibinfo{date}{March} \bibinfo{year}{2024}), \bibinfo{numpages}{25}~pages.
\newblock
\urldef\tempurl%
\url{https://doi.org/10.1145/3639264}
\showDOI{\tempurl}


\bibitem[\protect\citeauthoryear{Liu, {\"{O}}lveczky, Wang, Gupta, and Meseguer}{Liu et~al\mbox{.}}{2019a}]%
        {DBLP:journals/fac/LiuOWGM19}
\bibfield{author}{\bibinfo{person}{Si Liu}, \bibinfo{person}{Peter~Csaba {\"{O}}lveczky}, \bibinfo{person}{Qi Wang}, \bibinfo{person}{Indranil Gupta}, {and} \bibinfo{person}{Jos{\'{e}} Meseguer}.} \bibinfo{year}{2019}\natexlab{a}.
\newblock \showarticletitle{Read atomic transactions with prevention of lost updates: {ROLA} and its formal analysis}.
\newblock \bibinfo{journal}{\emph{Formal Aspects Comput.}} \bibinfo{volume}{31}, \bibinfo{number}{5} (\bibinfo{year}{2019}), \bibinfo{pages}{503--540}.
\newblock
\urldef\tempurl%
\url{https://doi.org/10.1007/S00165-019-00489-W}
\showDOI{\tempurl}


\bibitem[\protect\citeauthoryear{Liu, {\"{O}}lveczky, Zhang, Wang, and Meseguer}{Liu et~al\mbox{.}}{2019b}]%
        {ConsMaude:TACAS2019}
\bibfield{author}{\bibinfo{person}{Si Liu}, \bibinfo{person}{Peter~Csaba {\"{O}}lveczky}, \bibinfo{person}{Min Zhang}, \bibinfo{person}{Qi Wang}, {and} \bibinfo{person}{Jos{\'{e}} Meseguer}.} \bibinfo{year}{2019}\natexlab{b}.
\newblock \showarticletitle{Automatic Analysis of Consistency Properties of Distributed Transaction Systems in Maude}. In \bibinfo{booktitle}{\emph{{TACAS} 2019}} \emph{(\bibinfo{series}{LNCS})}, Vol.~\bibinfo{volume}{11428}. \bibinfo{publisher}{Springer}, \bibinfo{pages}{40--57}.
\newblock
\urldef\tempurl%
\url{https://doi.org/10.1007/978-3-030-17465-1\_3}
\showDOI{\tempurl}


\bibitem[\protect\citeauthoryear{MySQL}{MySQL}{2025}]%
        {MySQLTxnIsoLevel}
\bibfield{author}{\bibinfo{person}{MySQL}.} \bibinfo{year}{Accessed in December, 2025}\natexlab{}.
\newblock \bibinfo{title}{MySQL 8.4 Reference Manual - Transaction Isolation Levels}.
\newblock \bibinfo{howpublished}{\url{https://dev.mysql.com/doc/refman/8.4/en/innodb-transaction-isolation-levels.html}}.
\newblock


\bibitem[\protect\citeauthoryear{Papadimitriou}{Papadimitriou}{1979}]%
        {SER:JACM1979}
\bibfield{author}{\bibinfo{person}{Christos~H. Papadimitriou}.} \bibinfo{year}{1979}\natexlab{}.
\newblock \showarticletitle{The Serializability of Concurrent Database Updates}.
\newblock \bibinfo{journal}{\emph{J. ACM}} \bibinfo{volume}{26}, \bibinfo{number}{4} (\bibinfo{date}{Oct.} \bibinfo{year}{1979}), \bibinfo{pages}{631--653}.
\newblock
\urldef\tempurl%
\url{https://doi.org/10.1145/322154.322158}
\showDOI{\tempurl}


\bibitem[\protect\citeauthoryear{PostgreSQL}{PostgreSQL}{2025}]%
        {PostgreSQLTxnIsoLevel}
\bibfield{author}{\bibinfo{person}{PostgreSQL}.} \bibinfo{year}{Accessed in December, 2025}\natexlab{}.
\newblock \bibinfo{title}{PostgreSQL Documentation - Transaction Isolation}.
\newblock \bibinfo{howpublished}{\url{https://www.postgresql.org/docs/17/transaction-iso.html}}.
\newblock


\bibitem[\protect\citeauthoryear{Schultz and Demirbas}{Schultz and Demirbas}{2025}]%
        {10.14778/3750601.3750626}
\bibfield{author}{\bibinfo{person}{William Schultz} {and} \bibinfo{person}{Murat Demirbas}.} \bibinfo{year}{2025}\natexlab{}.
\newblock \showarticletitle{Design and Modular Verification of Distributed Transactions in MongoDB}.
\newblock \bibinfo{journal}{\emph{Proc. VLDB Endow.}} \bibinfo{volume}{18}, \bibinfo{number}{12} (\bibinfo{date}{Aug.} \bibinfo{year}{2025}), \bibinfo{pages}{5045–5058}.
\newblock
\showISSN{2150-8097}
\urldef\tempurl%
\url{https://doi.org/10.14778/3750601.3750626}
\showDOI{\tempurl}


\bibitem[\protect\citeauthoryear{Server}{Server}{2025}]%
        {SQLServerTxnIsoLevel}
\bibfield{author}{\bibinfo{person}{SQL Server}.} \bibinfo{year}{Accessed in December, 2025}\natexlab{}.
\newblock \bibinfo{title}{{{SET TRANSACTION ISOLATION LEVEL}} ({{Transact-SQL}})}.
\newblock \bibinfo{howpublished}{\url{https://learn.microsoft.com/en-us/sql/t-sql/statements/set-transaction-isolation-level-transact-sql?view=sql-server-ver17}}.
\newblock


\bibitem[\protect\citeauthoryear{Sovran, Power, Aguilera, and Li}{Sovran et~al\mbox{.}}{2011}]%
        {PSI:SOSP2011}
\bibfield{author}{\bibinfo{person}{Yair Sovran}, \bibinfo{person}{Russell Power}, \bibinfo{person}{Marcos~K. Aguilera}, {and} \bibinfo{person}{Jinyang Li}.} \bibinfo{year}{2011}\natexlab{}.
\newblock \showarticletitle{Transactional storage for geo-replicated systems}. In \bibinfo{booktitle}{\emph{SOSP '11}}. \bibinfo{pages}{385–400}.
\newblock
\urldef\tempurl%
\url{https://doi.org/10.1145/2043556.2043592}
\showDOI{\tempurl}


\bibitem[\protect\citeauthoryear{TiDB}{TiDB}{2025}]%
        {TiDBTxnIsoLevel}
\bibfield{author}{\bibinfo{person}{TiDB}.} \bibinfo{year}{Accessed in December, 2025}\natexlab{}.
\newblock \bibinfo{title}{{{TiDB Transaction Isolation Levels}}}.
\newblock \bibinfo{howpublished}{\url{https://docs.pingcap.com/tidb/stable/transaction-isolation-levels/}}.
\newblock


\bibitem[\protect\citeauthoryear{Xiong, Cerone, Raad, and Gardner}{Xiong et~al\mbox{.}}{2020}]%
        {xiongECOOP2020}
\bibfield{author}{\bibinfo{person}{Shale Xiong}, \bibinfo{person}{Andrea Cerone}, \bibinfo{person}{Azalea Raad}, {and} \bibinfo{person}{Philippa Gardner}.} \bibinfo{year}{2020}\natexlab{}.
\newblock \showarticletitle{Data {Consistency} in {Transactional Storage Systems}: {A Centralised Semantics}}. In \bibinfo{booktitle}{\emph{{ECOOP} 2020}}, Vol.~\bibinfo{volume}{166}. \bibinfo{publisher}{Schloss Dagstuhl -- Leibniz-Zentrum f\"ur Informatik}, \bibinfo{address}{Dagstuhl, Germany}, \bibinfo{pages}{21:1--21:31}.
\newblock
\showISBNx{978-3-95977-154-2}
\showISSN{1868-8969}
\urldef\tempurl%
\url{https://doi.org/10.4230/LIPIcs.ECOOP.2020.21}
\showDOI{\tempurl}


\bibitem[\protect\citeauthoryear{YugabyteDB}{YugabyteDB}{2025}]%
        {YugabyteDBTxnIsoLevel}
\bibfield{author}{\bibinfo{person}{YugabyteDB}.} \bibinfo{year}{Accessed in December, 2025}\natexlab{}.
\newblock \bibinfo{title}{YugabyteDB Docs - Isolation Levels}.
\newblock \bibinfo{howpublished}{\url{https://docs.yugabyte.com/preview/explore/transactions/isolation-levels/}}.
\newblock


\bibitem[\protect\citeauthoryear{Zhang, Sharma, Szekeres, Krishnamurthy, and Ports}{Zhang et~al\mbox{.}}{2018}]%
        {10.1145/3269981}
\bibfield{author}{\bibinfo{person}{Irene Zhang}, \bibinfo{person}{Naveen~Kr. Sharma}, \bibinfo{person}{Adriana Szekeres}, \bibinfo{person}{Arvind Krishnamurthy}, {and} \bibinfo{person}{Dan R.~K. Ports}.} \bibinfo{year}{2018}\natexlab{}.
\newblock \showarticletitle{Building Consistent Transactions with Inconsistent Replication}.
\newblock \bibinfo{journal}{\emph{ACM Trans. Comput. Syst.}} \bibinfo{volume}{35}, \bibinfo{number}{4}, Article \bibinfo{articleno}{12} (\bibinfo{date}{Dec.} \bibinfo{year}{2018}), \bibinfo{numpages}{37}~pages.
\newblock
\showISSN{0734-2071}
\urldef\tempurl%
\url{https://doi.org/10.1145/3269981}
\showDOI{\tempurl}


\bibitem[\protect\citeauthoryear{Zhuang, Lu, Liu, Chen, Shi, Zhao, Sun, Pan, and Du}{Zhuang et~al\mbox{.}}{2025}]%
        {TxnSails:VLDB2025}
\bibfield{author}{\bibinfo{person}{Qiyu Zhuang}, \bibinfo{person}{Wei Lu}, \bibinfo{person}{Shuang Liu}, \bibinfo{person}{Yuxing Chen}, \bibinfo{person}{Xinyue Shi}, \bibinfo{person}{Zhanhao Zhao}, \bibinfo{person}{Yipeng Sun}, \bibinfo{person}{Anqun Pan}, {and} \bibinfo{person}{Xiaoyong Du}.} \bibinfo{year}{2025}\natexlab{}.
\newblock \showarticletitle{{{TxnSails}}: {{Achieving Serializable Transaction Scheduling}} with {{Self-Adaptive Isolation Level Selection}}}. In \bibinfo{booktitle}{\emph{Proceedings of the {{VLDB Endowment}}}} \emph{(\bibinfo{series}{{{VLDB}}})}, Vol.~\bibinfo{volume}{18}. \bibinfo{pages}{4227--4240}.
\newblock
\showISSN{2150-8097}
\urldef\tempurl%
\url{https://doi.org/10.14778/3749646.3749689}
\showDOI{\tempurl}


\end{thebibliography}

\clearpage
\appendix
\section{Isolation Levels}
\label{app:isolation}

\noindent
\underline{\emph{Read Atomicity} ($\RA$)}~\cite{RAMP:TODS2016}
guarantees that if a transaction observes any update
from another transaction, it must observe all updates from that transaction.
In particular, it rules out \emph{fractured reads},
%
%
 e.g., 
Joey sees that Ross    has added Rachel   as a friend,
but fails to observe Rachel's  addition of Ross, 
resulting in an inconsistent view of a bi-directional friendship.

\smallskip
\noindent
\underline{\emph{Transactional Causal Consistency} ($\CC$)}~\cite{Cure:ICDCS2016}
 ensures that if a transaction $T'$ causally depends on another transaction $T$
(e.g., $T'$ reads a value written by $T$),
then any transaction that observes $T'$ must also observe $T$.
%
%
For example, under $\CC$, Ross observing the comment written by Rachel on Joey's post without seeing the post itself is not allowed.

\smallskip
\noindent
\underline{\emph{Prefix Consistency} ($\PC$)}~\cite{GSP:ECOOP2015}
strengthens $\CC$ by ensuring that concurrent transactions are not observed in different orders, which is illustrated by the  \emph{long-fork} anomaly:
\begin{align*}
	& T_1: \W(x, \mathit{1}) \quad T_3: \R(x, \mathit{1}), \R(y, \bot) \\
	& T_2: \W(y, \mathit{2}) \quad T_4: \R(x, \bot), \R(y, \mathit{2})
\end{align*}
where $T_3$ and $T_4$ observe incompatible prefixes of the commit order.
Equivalently,  under $\PC$, every transaction must observe a prefix of the global commit order.

\smallskip
\noindent
\underline{\emph{Parallel Snapshot Isolation} ($\PSI$)}~\cite{PSI:SOSP2011}
also strengthens $\CC$, but in a different way.
In particular, $\PSI$ forbids  \emph{lost updates}, where concurrent transactions overwrite each other's updates, e.g., when two transactions  deposit money into the same bank account at the same time, but one deposit is lost because it is overwritten by the other.
Unlike $\PC$, $\PSI$ still allows long-fork anomalies.

\smallskip
\noindent
\underline{\emph{Snapshot Isolation} ($\SI$)}~\cite{Critique:SIGMOD1995}
further enhances  both $\PSI$ and $\PC$ by ruling out
long forks  and lost updates.
In particular, in the earlier example, even if transactions $T_3$ and $T_4$
execute on different database replicas, they must observe a convergent snapshot
of the database state.
However, 
$\SI$ allows the 
\emph{write skew} anomaly:
\begin{align*}
	T_1&: \R(\mathit{a1}, 10), \R(\mathit{a2}, 10), \W(\mathit{a1}, -15) \\
	T_2&: \R(\mathit{a1}, 10), \R(\mathit{a2}, 10), \W(\mathit{a2}, -15),
\end{align*}
where both transactions check whether the total balance of the 
two bank accounts exceeds \$20 
before withdrawing \$15 from one of them.
When executed concurrently, both checks succeed, resulting in a negative total balance.

\smallskip
\noindent
\underline{\emph{Serializability} ($\SER$)}~\cite{SER:JACM1979}
forbids
all of the above anomalies.
It requires every concurrent transaction execution  to be equivalent to some serial order of those transactions.


\section{Generalizing the Framework of~\cite{Framework:CONCUR2015}}
\label{app-mil-il}


\begin{table*}[tbp]
  \centering
  	\captionsetup{skip=5pt}
  \caption{Consistency axioms  that constrain an abstract execution
		$\AE = (\T, \SO, \VIS, \AR)$.}
  \label{table:il-axioms}
  \renewcommand{\arraystretch}{1.3}
  \resizebox{\textwidth}{!}{%
    \begin{tabular}{|c|c|c|}
			\hline
			\multicolumn{3}{|c|}{
				$\begin{aligned}[c]
				\forall T \in \T.\; \forall r, x, v.\;
				(r = \R(x, v) \land \po^{-1}_{x}(r) \neq \emptyset)
				\implies \max_{\po}(\po^{-1}_{x}(r)) = \_(x,v).                                                                 \\
				\end{aligned}$ \hfill (\Int)} \\ \hline
			\multicolumn{3}{|c|}{
				$\begin{aligned}[c]
				\forall T \in \T.\; \forall x, v.\;
					T \vdash \R(x, v) \implies
					\max_{\AR}(\VIS^{-1}(T) \cap \WriteTx_{x}) \vdash \W(x, v)
				\end{aligned}$ \hfill (\Ext)}
			\\ \hline
			$\SO \subseteq \VIS$ \hfill (\Session)
			& $\VIS \comp \VIS \subseteq \VIS$ \hfill (\TransVis)
			& $\AR \comp \VIS \subseteq \VIS$ \hfill (\Prefixaxiom)
			\\ \hline
			$\forall T, T' \in \T.\;
				T \wwconflict T' \implies (T \rel{\VIS} T' \lor T' \rel{\VIS} T)$
				\hfill (\NoConflict)
			& \multicolumn{2}{c|}{
				$\AR \subseteq \VIS$ \hfill (\TotalVis)}
			\\ \hline
    \end{tabular}%
  }
\end{table*}


\begin{table}[tb]
	\centering
		\captionsetup{skip=5pt}
	\caption{Isolation levels defined by consistency axioms
	  in Table~\ref{table:il-axioms}.}
	\label{table:il-isolevel}
	\renewcommand{\arraystretch}{1.0}
	\begin{tabular}{p{0.45\columnwidth}|p{0.45\columnwidth}}
\toprule
		$\RA \equiv \Int \land \Ext \land \Session$ &
		$\CC \equiv \RA \land \TransVis$ \\  \midrule
		$\PC \equiv \RA \land \Prefixaxiom$ &
		$\PSI \equiv \RA \land \TransVis \land \NoConflict$ \\ \midrule
		$\SI \equiv \RA \land \Prefixaxiom \land \NoConflict$ &
		$\SER \equiv \RA \land \TotalVis$ \\
		\bottomrule
	\end{tabular}%
\end{table}

Homogeneous isolation settings correspond to the special case where all transactions are assigned the same isolation level.
In this section, we prove that 
\ourframework generalizes the framework of~\cite{Framework:CONCUR2015}. 
Tables~\ref{table:il-axioms} and~\ref{table:il-isolevel}
summarize the  isolation levels
and their corresponding consistency axioms defined in that framework. 


\begin{theorem}
	\label{thm:mil-il}
	For any history $\H = (\T, \SO, \level)$
	and traditional isolation level $\il$,
	$(\T, \SO) \models \il$ iff
	$\H$ is consistent under mixed isolation level
	where $\forall T \in \T.\; \level(T) = \il$.
\end{theorem}


\begin{proof}
	\label{proof:mil-il}
	We only need to show that for any abstract execution
	$\AE = (\T, \SO, \level, \VIS, \AR)$,
	$\AE' \triangleq (\T, \SO, \VIS, \AR) \models \NoConflict$
	if and only if $\forall T \in \T.\; \AE \models \NoConflict(T)$
	under the condition that
	$\level(T) \in \set{\PSI, \SI, \SER}$ for all $T \in \T$.
	Other consistency axioms correspond directly
	between mixed isolation levels and traditional isolation levels
	and thus the proof is omitted here.
	\begin{itemize}
		\item The ``$\implies$'' direction.
		  Suppose that $\AE' \models \NoConflict$.
		  Consider any transactions $T, T' \in \T$
			such that $T' \wwconflict T$ and $\NoConflict(T)$ holds for $T$.
			Now suppose that $T' \rel{\AR} T$.
			We need to show that $T' \rel{\VIS} T$ holds.
			Suppose by contradiction that
			\[
				\lnot(T' \rel{\VIS} T).
			\]
		  Since $\AE' \models \NoConflict$, we have
			\[
				T' \rel{\VIS} T \lor T \rel{\VIS} T'.
			\]
			Therefore, we have
			\[
				T \rel{\VIS} T'.
			\]
			Since $\VIS \subseteq \AR$, we have
			\[
				T \rel{\AR} T'.
			\]
			This contradicts the assumption that $T' \rel{\AR} T$.
			Hence, $T' \rel{\VIS} T$ holds,
			and thus $\AE \models \NoConflict(T)$.
		\item The ``$\impliedby$'' direction.
		  Suppose that $\AE \models \NoConflict(T)$ for all $T \in \T$
			with $\level(T) \in \set{\PSI, \SI, \SER}$.
		  Consider any transactions $T, T' \in \T$
			such that $T' \wwconflict T$.
			We need to show that
			\[
				T' \rel{\VIS} T \lor T \rel{\VIS} T'.
			\]
			Suppose by contradiction that
			\[
				\lnot(T' \rel{\VIS} T) \land \lnot(T \rel{\VIS} T').
			\]
		  Since $\forall T \in \T.\; \AE \models \NoConflict(T)$,
			we have that
			\[
				\AE \models \NoConflict(T) \land \AE \models \NoConflict(T').
			\]
			Suppose without loss of generality that $T' \rel{\AR} T$.
			Then, by $\AE \models \NoConflict(T)$, we have
			\[
				T' \rel{\VIS} T,
			\]
			which contradicts the assumption.
			Hence, $T' \rel{\VIS} T \lor T \rel{\VIS} T'$ holds,
			and thus $\AE' \models \NoConflict$.
	\end{itemize}
\end{proof}



\section{Equivalence between \ourframework and the Framework of~\cite{ComplexityMIL:CAV2025} }
\label{app-mil-cav-equiv}

Recently, Bouajjani et al.~\cite{ComplexityMIL:CAV2025}
 proposed an axiomatic semantics for mixed isolation levels.
To validate our \ourframework, 
we  establish its equivalence to that semantics 
on their common isolation levels, 
namely $\RA$, $\PC$, $\SI$, and $\SER$.

\begin{definition}[Histories in~\cite{ComplexityMIL:CAV2025}]
	\label{def:cav-history}
	A \bfit{history} $\h = (\T, \SO, \level, \WR)$
	is a set $\T$ of transactions
	along with a strict partial session order $\SO \subseteq \T \times \T$,
	an allocation function $\level: \T \to \IL$,
	and a write-read relation $\WR: \Key \to 2^{\T \times \T}$ such that
	($\exists!$ means ``unique existence''):
	\begin{itemize}
		\item $\forall x \in \Key, S \in \T.\;
					S \vdash \R(x, \_) \implies
						\exists!\; T \in \T.\; T \rel{\WR(x)} S$.
		\item $\forall x \in \Key.\; \forall T, S \in \T.\;
					T \rel{\WR(x)} S \implies
					\exists v \in \Val.\; T \neq S \land
					T \vdash \W(x, v) \land
					S \vdash \W(x, v)$.
	  \item $(\SO \cup \WR)$ is acyclic.
	\end{itemize}
\end{definition}

For notational simplicity, Bouajjani et al.'s framework
makes the following standard assumptions about histories.
\begin{itemize}
	\item \OneWrite: Each transaction contains at most one write operation per key.
	\item \RYW: A read operation preceded by write operations on the same key
	  returns the value written by the last preceding write to this key in the transaction.
	\item \InitTran: Each history contains a special transaction
		that writes initial value $\init$ to all keys.
		This transaction precedes all the other transactions in $\SO$.
\end{itemize}

\begin{definition}[Abstract Executions in~\cite{ComplexityMIL:CAV2025}]
	\label{def:cav-ae}
	An \bfit{abstract execution}
	$\aecav = (\T, \SO, \level, \WR, \CO)$
	is a history $\h = (\T, \SO, \level, \WR)$
	along with a strict total order $\CO \subseteq \T \times \T$
	called commit order such that $(\SO \cup \WR) \subseteq \CO$.
\end{definition}


\begin{figure*}[t]
	\centering
	\begin{minipage}{\textwidth}
		\begin{align*}
			\ReadAtomic(T) \equiv\;
					&\forall x \in \Key.\; \forall T_1, T_2 \in \T.\; \\
						&\quad T_1 \neq T_2 \land T_1 \rel{\WR(x)} T \land
						T_{2} \in \WriteTx_{x} \land \\
						&\quad T_2 \rel{\SO \;\cup\; \WR} T
							\implies T_2 \rel{\CO} T_{1}. \\
			\Prefix(T) \equiv\;
				&\forall x \in \Key.\; \forall T_1, T_2 \in \T.\; \\
					&\quad T_1 \neq T_2 \land T_1 \rel{\WR(x)} T
					\land T_{2} \in \WriteTx_{x} \land \\
				  &\quad T_2 \rel{\CO^{\ast} \comp (\SO \;\cup\; \WR)} T
					  \implies T_2 \rel{\CO} T_1. \\
			\Conflict(T) \equiv\;
				&\forall x \in \Key.\; \forall T_1, T_2, T_{3} \in \T.\; \\
					&\quad T_1 \neq T_2 \land T_1 \rel{\WR(x)} T \land
					T_{2} \in \WriteTx_{x} \land \\
					&\quad T_{3} \wwconflict T \land
						T_2 \rel{\CO^{\ast}} T_3 \rel{\CO} T
						\implies T_2 \rel{\CO} T_1. \\
			\SnapshotIsolation(T) \equiv\; & \Prefix(T) \land \Conflict(T). \\
			\Serializability(T) \equiv\;
				&\forall x \in \Key.\; \forall T_1, T_2 \in \T.\;
					T_1 \neq T_2 \land T_1 \rel{\WR(x)} T
					\land T_{2} \in \WriteTx_{x} \land
					T_2 \rel{\CO} T \implies T_2 \rel{\CO} T_1.
		\end{align*}
	\end{minipage}
	\caption{Consistency axioms for individual transactions
	  in the framework of Bouajjani et al.~\cite{ComplexityMIL:CAV2025}.}
	\label{fig:cav-axioms}
\end{figure*}

\begin{definition}[Consistency Axioms for Individual Transactions in~\cite{ComplexityMIL:CAV2025}]
	\label{def:cav-axioms}
	Let $\aecav = (\T, \SO, \level, \WR, \CO)$ be an abstract execution
	and $T \in \T$ be a transaction.
	The \bfit{consistency axioms} for individual transactions
	are defined as in Figure~\ref{fig:cav-axioms}.
\end{definition}

Read Atomic ($\ReadAtomic(T)$), Prefix ($\Prefix(T)$),
and Serializability ($\Serializability(T)$) axioms
are defined using their homonymous axioms,
and Snapshot Isolation ($\SnapshotIsolation(T)$) is defined
as a conjunction of $\Prefix(T)$ and $\Conflict(T)$.
Note that the framework of Bouajjani et al.
does not include the $\Causal(T)$ axiom and,
consequently, does not capture the $\CC$ or $\PSI$ isolation levels.
Therefore, in the following equivalence theorem,
we only consider the isolation levels in $\set{\RA, \PC, \SI, \SER}$.

The notion of ``consistent abstract executions''
and ``consistent histories'' are defined
in the same way as in our framework.

\begin{definition}[Consistent Histories in~\cite{ComplexityMIL:CAV2025}]
	\label{def:cav-consistent}
	An {\it abstract execution}
	$\aecav = (\T, \SO, \level, \WR, \CO)$ is called \bfit{consistent}
	if for each transaction $T \in \T$,
	the consistency axioms corresponding to its isolation level
	$\level(T)$ hold on $T$.

	A {\it history} $\h$ is called \bfit{consistent} if
	there exists a consistent abstract execution
	$\aecav = (\h, \WR, \CO)$ for $\h$.
\end{definition}

The equivalence between our framework
and that of Bouajjani et al.~\cite{ComplexityMIL:CAV2025}
is established as follows.


\begin{theorem}
	\label{thm:framework-equiv}
	Let $\h = (\T, \SO, \level, \WR)$ be a history
	in Bouajjani et al.'s framework
	and $\H = (\T, \SO, \level)$ be the corresponding history
	in our framework.
	Then, $\h$ is consistent iff $\H$ is consistent.
\end{theorem}



\begin{proof}[Proof of Theorem~\ref{thm:framework-equiv}]
	Suppose that $h = (\T, \SO, \level, \WR)$
	is a history in Bouajjani et al.'s framework,
	and $\H = (\T, \SO, \level)$ is the corresponding history in our framework.


\begin{figure*}[t]
	\centering
	\begin{minipage}{1.0\textwidth}
		\begin{align}
			\forall S, &T \in \T.\; S \rel{\VIS} T \iff \nonumber \\
				& (\level(T) = \RA \land S \rel{\SO \;\cup\; \WR} T) \;\lor \\
				& (\level(T) = \PC \land S \rel{\CO^{\ast} \comp (\SO \;\cup\; \WR)} T) \;\lor \\
				& (\level(T) = \SI \land S \rel{
					\CO^{\ast} \comp (\SO \;\cup\; \WR \;\cup\; \set{
						(T_{1}, T_{2}) \;\mid\; T_{1} \;\rel{\CO}\; T_{2} \;\land\;
							T_{1} \;\wwconflict\; T_{2}})} T) \;\lor \\
				& (\level(T) = \SER \land S \rel{\CO} T).
		\end{align}
	\end{minipage}
	\caption{Definition of $\VIS$ in terms of $\SO$, $\WR$, and $\CO$
	  in the framework of Bouajjani et al.~\cite{ComplexityMIL:CAV2025}.}
	\label{fig:vis-cav}
\end{figure*}


	On the one hand, we show that if $\h$ is consistent, then $\H$ is consistent.
	Since $\h$ is consistent, by Definition~\ref{def:cav-consistent},
	there exists an abstract execution
	$\aecav = (\T, \SO, \level, \WR, \CO)$
	such that for every transaction $T \in \T$,
	the consistency axioms corresponding to its isolation level
	$\level(T)$ hold on $T$.
	Now,
	we need to construct an abstract execution
	$\AE = (\H, \VIS, \AR)$
	such that for every transaction $T \in \T$,
	the consistency axioms corresponding to its isolation level
	$\level(T)$ in our framework hold on $T$.
	To this end, we define $\AR = \CO$ and $\VIS$
	as in Figure~\ref{fig:vis-cav}.

	For $\AE$ to be a valid abstract execution,
	we need to show that
	\begin{itemize}
		\item $\AR$ is a strict total order.
		  This holds since $\AR = \CO$,
			which is a strict total order by Definition~\ref{def:cav-ae}.
		\item $\VIS \subseteq \AR$.
			We distinguish cases based on the isolation level of $T$
			and its construction of $\VIS$,
			and show that $\VIS \subseteq \AR$ holds in each case.
			\begin{itemize}
				\item \case{} $\level(T) = \RA$.
				  By the construction of $\VIS$ for $\level(T) = \RA$ in Figure~\ref{fig:vis-cav},
					we need to show that
					\[
						(\SO \cup \WR) \subseteq \AR.
					\]
					This holds due to $(\SO \cup \WR) \subseteq \CO$
					(Definition~\ref{def:cav-ae})
					and the fact that $\AR = \CO$.
				\item \case{} $\level(T) = \PC$.
				  By the construction of $\VIS$ for $\level(T) = \PC$
					in Figure~\ref{fig:vis-cav},
					we need to show that
					\[
						\CO^{\ast} \comp (\SO \cup \WR) \subseteq \AR.
					\]
					This holds since $(\SO \cup \WR) \subseteq \AR$
					by \case{} $\level(T) = \RA$
					and $\AR = \CO$ is a strict total order.
				\item \case{} $\level(T) = \SI$.
				  By the construction of $\VIS$ for $\level(T) = \SI$ in Figure~\ref{fig:vis-cav},
					we need to show that
										\[
												\CO^{\ast} \comp (\SO \cup \WR \cup
													\set{(T_{1}, T_{2}) \mid T_{1} \rel{\CO} T_{2} \land
													T_{1} \wwconflict T_{2}}) \subseteq \AR.
										\]
					This holds since $(\SO \cup \WR) \subseteq \AR$
					by \case{} $\level(T) = \RA$
					and $\AR = \CO$ is a strict total order.
				\item \case{} $\level(T) = \SER$.
				  By the construction of $\VIS$ for $\level(T) = \SER$ in Figure~\ref{fig:vis-cav},
					we need to show that
					\[
						\CO \subseteq \AR.
					\]
					This holds since $\AR = \CO$.
			\end{itemize}
		\item $\VIS$ is acyclic.
		  This holds since $\VIS \subseteq \AR$ and $\AR$ is acyclic.
	\end{itemize}

	In the following, we show that for every transaction $T \in \T$,
	the consistency axioms of our framework
	corresponding to its isolation level $\level(T)$ hold on $T$.

	\begin{itemize}
		\item \case{} $\level(T) = \RA$.
		  Since $\aecav$ is consistent, $\ReadAtomic(T)$ holds.
			\begin{itemize}
				\item $\Int(T)$.
				  Consider any internal read $r = \R(x, v)$ in $T$
					for some $x$ and $v$.
					If $r$ does not exists, $\Int(T)$ holds trivially.
					Otherwise, let $o \triangleq \max_{\po}(\po_{x}^{-1}(r))$
					be the last operation on $x$ before $r$ in $T$.
					We show that $o = \_(x, v)$,
					by distinguishing two cases depending on the type of $o$:
					\begin{itemize}
						\item $o$ is a write operation.
						  By the $\RYW$ assumption about histories,
							$r$ reads the value written by $o$.
							Thus, $o = \W(x, v)$.
						\item $o$ is a read operation.
						  By the choise of $o$, there is no write operation on $x$
							between $o$ and $r$ in $T$.
							We then distinguish two sub-cases depending on
							whether there is a write operation on $x$ in $T$ {\it before} $o$.
							First, suppose there is such a write operation,
							then by $\OneWrite$ and $\RYW$,
							both $o$ and $r$ read the value written by that write operation,
							Since $r = \R(x, v)$, we have $o = \R(x, v)$ as well.
							Second, suppose there is no such a write operation,
						  we show that both $o$ and $r$ must read from
							the same transaction.
							Suppose not, i.e., there exists transactions $T_{1} \neq T_{2}$
							such that $T_{1} \rel{\WR(x)} T$ and $T_{2} \rel{\WR(x)} T$.
							By $\ReadAtomic(T)$, we have
							$T_{1} \rel{\CO} T_{2} \land T_{2} \rel{\CO} T_{1}$.
							Therefore, $\CO$ is cyclic,
							contradicting the fact that $\CO$ is a strict total order.
							Thus, there exists a unique transaction $T'$ such that
							both $o$ and $r$ read from $T'$.
							By $\OneWrite$ for $T'$ and $r = \R(x, v)$,
							we have $o = \R(x, v)$.
					\end{itemize}
				\item $\Ext(T)$.
				  Consider any external read $r = \R(x, v)$ in $T$
					for some $x$ and $v$.
					If $r$ does not exists, $\Ext(T)$ holds trivially.
					Otherwise, let $T_{1}$ be the unique transaction
					such that
					\[
						T_{1} \vdash \W(x, v) \land T_{1} \rel{\WR(x)} T.
					\]
					Hence, by the construction of $\VIS$ for $\level(T) = \RA$,
					\[
						T_{1} \in (\VIS^{-1}(T) \cap \WriteTx_{x}).
					\]
					Consider any $T_{2} \in \T$ such that
					\[
						T_{1} \neq T_{2} \land T_{2} \in (\VIS^{-1}(T) \cap \WriteTx_{x}).
					\]
					That is,
					\[
						T_{1} \neq T_{2} \land T_{2} \in \WriteTx_{x} \land
						  T_{2} \rel{\SO \;\cup\; \WR} T.
					\]
					By $\ReadAtomic(T)$, we have
					\[
						T_{2} \rel{\CO} T_{1}.
					\]
					That is,
					\[
						T_{2} \rel{\AR} T_{1}.
					\]
					Therefore,
					\[
						T_{1} = \max_{\AR}(\VIS^{-1}(T) \cap \WriteTx_{x}).
					\]
					Hence, $\Ext(T)$ holds.
				\item $\Session(T)$.
				  Consider any $T' \in \T$ such that $T' \rel{\SO} T$.
					By construction of $\VIS$ for $\level(T) = \RA$,
					we have $T' \rel{\VIS} T$.
			\end{itemize}
		\item \case{} $\level(T) = \PC$.
			$\Int(T)$, $\Ext(T)$, and $\Session(T)$ hold
			by similar reasoning as in the case of $\level(T) = \RA$.
			In the following, we show that $\Prefixaxiom(T)$ holds.
			Consider any $T_{1}, T_{2} \in \T$ such that
			\[
				T_{1} \rel{\AR} T_{2} \rel{\VIS} T.
			\]
			By the construction of $\VIS$ and $\AR$ for $\level(T) = \PC$,
			\[
				T_{1} \rel{\CO} T_{2} \rel{\CO^{\ast} \comp (\SO \;\cup\; \WR)} T.
			\]
			That is,
			\[
				\exists T_{3} \in \T.\;
				T_{1} \rel{\CO} T_{2} \rel{\CO^{\ast}} T_{3} \rel{\SO \;\cup\; \WR} T.
			\]
			Since $\CO$ is transitive, we have
			\[
				T_{1} \rel{\CO} T_{3} \rel{\SO \;\cup\; \WR} T.
			\]
			Hence, by the construction of $\VIS$ for $\level(T) = \PC$,
			\[
				T_{1} \rel{\VIS} T.
			\]
			Therefore, $\Prefixaxiom(T)$ holds.
		\item \case{} $\level(T) = \SI$.
			$\Int(T)$, $\Ext(T)$, $\Session(T)$, and $\Prefixaxiom(T)$ hold
			by similar reasoning as in the case of $\level(T) = \PC$.
			In the following, we show that $\NoConflict(T)$ holds.
			Consider any $T' \in \T$ such that
			\[
				T' \wwconflict T \land T' \rel{\AR} T.
			\]
			By the construction of $\AR$ for $\level(T) = \SI$,
			\[
				T' \wwconflict T \land T' \rel{\CO} T.
			\]
			That is,
			\[
				T' \rel{\set{(T_{1}, T_{2}) \;\mid\;
				  T_{1} \;\rel{\CO}\; T_{2} \;\land\; T_{1} \;\wwconflict\; T_{2}}} T.
			\]
			By the construction of $\VIS$ for $\level(T) = \SI$,
			\[
				T' \rel{\VIS} T.
			\]
			Therefore, $\NoConflict(T)$ holds.
		\item \case{} $\level(T) = \SER$.
			$\Int(T)$, $\Ext(T)$, and $\Session(T)$ hold
			by similar reasoning as in the case of $\level(T) = \RA$.
			In the following, we show that $\TotalVis(T)$ holds.
			Consider any $T' \in \T$ such that
			\[
				T' \rel{\AR} T.
			\]
			By the construction of $\AR$ for $\level(T) = \SER$,
			\[
				T' \rel{\CO} T.
			\]
			By the construction of $\VIS$ for $\level(T) = \SER$,
			\[
				T' \rel{\VIS} T.
			\]
			Therefore, $\TotalVis(T)$ holds.
	\end{itemize}

	On the other hand, we show that if $\H$ is consistent,
	then $\h$ is consistent.
	Since $\H$ is consistent, 
	there exists an abstract execution $\AE = (\H, \VIS, \AR)$
	such that for every transaction $T \in \T$,
	the consistency axioms corresponding to its isolation level
	$\level(T)$ in our framework hold on $T$.
	By Definition~\ref{def:cav-consistent},
	we need to construct an abstract execution
	$\aecav = (\T, \SO, \level, \WR, \CO)$
	such that for every transaction $T \in \T$,
	the consistency axioms corresponding to its isolation level
	$\level(T)$ in Bouajjani et al.'s framework hold on $T$.
	To this end, we define $\CO = \AR$, which is a strict total order.

	By $\Ext(T)$, we define $\WR(x)$ like:
	\[
		T' \rel{\WR(x)} T \iff T' = \max_{\AR}(\VIS^{-1}(T) \cap \WriteTx_{x}).
	\]
	and then $\WR = \bigcup_{x \in \Key} \WR(x)$.

	For $\aecav$ to be a valid abstract execution,
	we need to show that
	\[
		(\SO \cup \WR) \subseteq \CO.
	\]
	First, by the definition of $\WR(x)$ above, we have
	\[
		\WR \subseteq \VIS.
	\]
	Since $\VIS \subseteq \AR$ and $\CO = \AR$, we have
	\[
		\WR \subseteq \CO.
	\]
	Since $\ae$ is an abstract execution and
	all consistency levels we consider enforces $\SO(T)$, we have
	\[
		\SO \subseteq \AR = \CO.
	\]
	Therefore, we have
	\[
		(\SO \cup \WR) \subseteq \CO.
	\]
	In the following, we show that for every transaction $T \in \T$,
	the consistency axioms of the framework of Bouajjani et al.
	corresponding to its isolation level $\level(T)$ hold on $T$.
	\begin{itemize}
		\item \case{} $\level(T) = \RA$.
		  Since $\H$ is consistent,
			$\Int(T)$, $\Ext(T)$, and $\Session(T)$ hold.
			We need to show that $\ReadAtomic(T)$ holds.
			Consider any $T_1, T_2 \in \T$ such that for some $x$,
			\[
				T_1 \neq T_2 \land T_1 \rel{\WR(x)} T \land
				  T_2 \in \WriteTx_{x} \land T_2 \rel{\SO \;\cup\; \WR} T.
			\]
			{By $\Ext(T)$ and $T_{1} \rel{\WR(x)} T$, we have}
			\[
				T_1 = \max_{\AR}(\VIS^{-1}(T) \cap \WriteTx_{x}).
			\]
			By $\Session(T)$, $\Ext(T)$,
			and $T_2 \rel{\SO \;\cup\; \WR} T$, we have
			\[
				T_2 \rel{\VIS} T.
			\]
			Since $T_2 \in \WriteTx_{x}$,
			\[
				T_2 \in (\VIS^{-1}(T) \cap \WriteTx_{x}).
			\]
			Thus,
			\[
				T_2 \rel{\AR} T_1.
			\]
			Since $\CO = \AR$,
			\[
				T_2 \rel{\CO} T_1.
			\]
			Thus, $\ReadAtomic(T)$ holds.
		\item \case{} $\level(T) = \PC$.
		  Since $\H$ is consistent, $\Int(T)$, $\Ext(T)$, $\Session(T)$,
			and $\Prefixaxiom(T)$ hold.
			We need to show that $\Prefix(T)$ holds.
			Consider any $T_1, T_2 \in \T$ such that for some $x$,
			\[
				T_1 \neq T_2 \land T_1 \rel{\WR(x)} T \land
				  T_2 \in \WriteTx_{x} \land
				  \exists T_{3}.\;
					  T_2 \rel{\CO^{\ast}} T_{3} \rel{\SO \;\cup\; \WR} T.
			\]
			{By $\Ext(T)$ and $T_{1} \rel{\WR(x)} T$, we have}
			\[
				T_1 = \max_{\AR}(\VIS^{-1}(T) \cap \WriteTx_{x}).
			\]
			By $\Session(T)$, $\Ext(T)$,
			and $T_{3} \rel{\SO \;\cup\; \WR} T$,
			\[
				T_{3} \rel{\VIS} T.
			\]
			By $\Prefixaxiom(T)$ and $\CO = \AR$,
			\[
				T_{2} \rel{\CO} T_{3} \rel{\VIS} T
				  \implies T_{2} \rel{\VIS} T.
			\]
			Therefore,
			\[
				T_{2} \rel{\CO^{\ast}} T_{3} \rel{\VIS} T
				  \implies T_{2} \rel{\VIS} T.
			\]
			Since $T_2 \in \WriteTx_{x}$,
			\[
				T_2 \in (\VIS^{-1}(T) \cap \WriteTx_{x}).
			\]
			Therefore,
			\[
				T_2 \rel{\AR} T_1.
			\]
			Since $\CO = \AR$,
			\[
				T_2 \rel{\CO} T_1.
			\]
			Thus, $\Prefix(T)$ holds.
		\item \case{} $\level(T) = \SI$.
		  Since $\H$ is consistent, $\Int(T)$, $\Ext(T)$, $\Session(T)$,
			$\Prefixaxiom(T)$, and $\NoConflict(T)$ hold.
			We need to show that both $\Prefix(T)$ and $\Conflict(T)$ hold.
			Since $\SI(T) \equiv \PC(T) \land \NoConflict(T)$
			and we have shown that $\Prefix(T)$ holds if $\PC(T)$ holds,
			it suffices to show that $\Conflict(T)$ holds.
			Consider any $T_1, T_2, T_{3} \in \T$ such that for some $x$,
			\[
				T_1 \neq T_2 \land T_1 \rel{\WR(x)} T \land
				  T_{2} \in \WriteTx_{x} \land
				  T_{3} \wwconflict T \land
				  T_2 \rel{\CO^{\ast}} T_{3} \rel{\CO} T.
			\]
			{By $\Ext(T)$ and $T_{1} \rel{\WR(x)} T$, we have}
			\[
				T_1 = \max_{\AR}(\VIS^{-1}(T) \cap \WriteTx_{x}).
			\]
			By $\NoConflict$, $T_{3} \wwconflict T$, $T_{3} \rel{\CO} T$,
			and $\CO = \AR$,
			\[
				T_{3} \rel{\VIS} T.
			\]
			By $\Prefixaxiom(T)$ and $\CO = \AR$,
			\[
				T_{2} \rel{\CO} T_{3} \rel{\VIS} T
				  \implies T_{2} \rel{\VIS} T.
			\]
			Therefore,
			\[
				T_{2} \rel{\CO^{\ast}} T_{3} \rel{\VIS} T
				  \implies T_{2} \rel{\VIS} T.
			\]
			Since $T_2 \in \WriteTx_{x}$,
			\[
				T_2 \in (\VIS^{-1}(T) \cap \WriteTx_{x}).
			\]
			Therefore,
			\[
				T_2 \rel{\AR} T_1.
			\]
			Since $\CO = \AR$,
			\[
				T_2 \rel{\CO} T_1.
			\]
			Thus, $\Conflict(T)$ holds.
		\item \case{} $\level(T) = \SER$.
		  Since $\H$ is consistent, $\Int(T)$, $\Ext(T)$, $\Session(T)$,
			and $\TotalVis(T)$ hold.
			We need to show that $\Serializability(T)$ holds.
			Consider any $T_1, T_2 \in \T$ such that for some $x$,
			\[
				T_1 \neq T_2 \land T_1 \rel{\WR(x)} T \land
				  T_2 \in \WriteTx_{x} \land T_2 \rel{\CO} T.
			\]
			{By $\Ext(T)$ and $T_{1} \rel{\WR(x)} T$, we have}
			\[
				T_1 = \max_{\AR}(\VIS^{-1}(T) \cap \WriteTx_{x}).
			\]
			By $\TotalVis(T)$, $T_2 \rel{\CO} T$, and $\CO = \AR$,
			\[
				T_2 \rel{\VIS} T.
			\]
			Since $T_2 \in \WriteTx_{x}$,
			\[
				T_2 \in (\VIS^{-1}(T) \cap \WriteTx_{x}).
			\]
			Therefore,
			\[
				T_2 \rel{\AR} T_1.
			\]
			Since $\CO = \AR$,
			\[
				T_2 \rel{\CO} T_1.
			\]
			Thus, $\Serializability(T)$ holds.
	\end{itemize}
\end{proof}

\section{The PC--SI--SER Protocol}
\label{app-pc-si-ser}


\begin{algorithm}[t]
  \small
  \caption{The \PC-\SI-\SER{} protocol}
  \label{alg:pc-si-ser}
  \begin{algorithmic}[1]
		\Statex $\store$ and $T.\buffer$ as in Algorithm~\ref{alg:si-ser-fekete}

		\hStatex
    \Procedure{\StartProc}{$T$}
      \label{line:proc-starttxn-pcsiser}
			\State $T.\sts \gets \now$
			  \Comment{also for \SER{} transactions}
			  \label{line:starttxn-starttime-pcsiser}
    \EndProcedure

		\hStatex
    \Procedure{\WriteProc}{$T, k, v$}
      \label{line:proc-writetxn-pcsiser}
			\State $T.\buffer[k] \gets v$
			  \label{line:writetxn-buffer-pcsiser}
    \EndProcedure

		\hStatex
    \Procedure{\ReadProc}{$T, k$}
      \label{line:proc-readtxn-pcsiser}
			\If{$k \in \dom(T.\buffer)$}
				\State \Return $T.\buffer[k]$
				\label{line:readtxn-buffer-pcsiser}
			\EndIf
			\State \Return value of $\store[k]$ at latest timestamp $< T.\sts$
			\label{line:readtxn-store-pcsiser}
    \EndProcedure

		\hStatex
    \Procedure{\CommitProc}{$T$}
      \label{line:proc-committxn-pcsiser}
			\State $T.\cts \gets \now$
			  \label{line:committxn-committime-pcsiser}
			\If{$\level(T) = \SI \land
			  \exists T'.\; T'.\cts \in (T.\sts, T.\cts) \land
				T' \wwconflict T$}
			  \State \Return $\aborted$
					\label{line:committxn-si-conflict-pcsiser} \label{line:aborttxn-status-pcsiser}
			\EndIf
			\If{$\level(T) = \SER \land
			  \exists T'.\; T'.\cts \in (T.\sts, T.\cts) \land
				(T \wwconflict T' \lor T \rwconflict T')$}
			  \State \Return $\aborted$ \label{line:committxn-ser-conflict-pcsiser}
			\EndIf
			\State $\store[k] \gets [T.\cts \mapsto v], \forall [k \mapsto v] \in T.\buffer$
				\label{line:committxn-store-pcsiser}
			\State \Return $\committed$
			  \label{line:committxn-return-status-pcsiser}
    \EndProcedure
  \end{algorithmic}
\end{algorithm}

We propose a new concurrency control protocol called \PC-\SI-\SER,
which supports clients choosing among $\PC$, $\SI$, and $\SER$ isolation levels
for each transaction.
To the best of our knowledge, this is the first protocol that mixes three isolation levels
spanning from the weakest (\PC) to the strongest (\SER) in a unified framework.
Algorithm~\ref{alg:pc-si-ser} presents its pseudocode,
which adapts the well-known centralized $\SI$ protocol~\cite{Critique:SIGMOD1995,Adya:PhDThesis1999}
by extending its conflict-checking logic to distinguish among the three levels.

All three levels follow the same basic flow:
at begin, $T$ obtains a start timestamp $T.\sts$
(\code{\ref{alg:pc-si-ser}}{\ref{line:starttxn-starttime-pcsiser}});
during execution, writes are buffered in $T.\buffer$
(\code{\ref{alg:pc-si-ser}}{\ref{line:writetxn-buffer-pcsiser}}),
and reads return values from $T.\buffer$ for keys already written by $T$
(\code{\ref{alg:pc-si-ser}}{\ref{line:readtxn-buffer-pcsiser}}),
or from the snapshot of $\store$ as of $T.\sts$ for external reads
(\code{\ref{alg:pc-si-ser}}{\ref{line:readtxn-store-pcsiser}});
at commit, $T$ obtains a commit timestamp $T.\cts$ and installs its writes
(\code{\ref{alg:pc-si-ser}}{\ref{line:committxn-committime-pcsiser}}).
The levels differ only in conflict detection:
$\PC$ performs no conflict checks;
$\SI$ checks for write-write conflicts with concurrent transactions
($T'.\cts \in (T.\sts, T.\cts)$ and $T' \wwconflict T$) and aborts if any are found
(\code{\ref{alg:pc-si-ser}}{\ref{line:committxn-si-conflict-pcsiser}});
$\SER$ additionally checks for read-write conflicts
($T \rwconflict T'$), aborting if a concurrent transaction has written to any key that $T$ read
(\code{\ref{alg:pc-si-ser}}{\ref{line:committxn-ser-conflict-pcsiser}}).
This graduated scheme captures the increasing strength from $\PC$ through $\SI$ to $\SER$.


\begin{theorem}
	\label{thm:pc-si-ser}
	Algorithm~\ref{alg:pc-si-ser} conforms to the mixed isolation guarantees of \PC, \SI, and \SER.
\end{theorem}



\begin{proof}
	For any history $\H = (\T, \SO, \level)$ produced by Algorithm~\ref{alg:pc-si-ser},
	we define $\AR$ as follows:
	\[
		\forall T', T \in \T.\;
		T' \rel{\AR} T \iff T'.\cts < T.\cts.
	\]
	Clearly, $\AR$ is a strict total order.
	The visibility relation $\VIS$ is defined as
	\begin{align*}
		\forall T', T \in \T.\; & T' \rel{\VIS} T \iff                            \\
		                        & (\level(T) = \PC \land T'.\cts < T.\sts) \;\lor \\
		                        & (\level(T) = \SI \land T'.\cts < T.\sts) \;\lor \\
		                        & (\level(T) = \SER \land T'.\cts < T.\cts).
	\end{align*}
	It is easy to verify that $\VIS$ is irreflexive
	and $\VIS \subseteq \AR$.~\footnote{
		The proof is similar to that in the proof of Theorem~\ref{thm:si-ser-fekete}.
		We omit the details for brevity.
	}
	In the following, we show that the corresponding abstract execution
	$\AE = (\H, \VIS, \AR)$ for $\H$ is consistent:
	Fix a transaction $T \in \T$.
	We show that the consistency axioms
	corresponding to $\level(T)$ hold on $T$.

	\begin{itemize}
		\item $\Int(T)$:
		      If there is no \textit{internal} read in $T$,
		      then $\Int(T)$ holds vacuously.
		      Otherwise, $\Int(T)$ holds
		      because $T$ buffers writes in its private $\buffer$
		      (\code{\ref{alg:pc-si-ser}}{\ref{line:writetxn-buffer-pcsiser}})
		      and reads from $\buffer$ for internal reads
		      (\code{\ref{alg:pc-si-ser}}{\ref{line:readtxn-buffer-pcsiser}}).
		\item $\Ext(T)$:
		      If there is no \textit{external} read in $T$,
		      then $\Ext(T)$ holds vacuously.
		      Otherwise, consider any external read $r \triangleq \R(x, v)$ in $T$
		      for some $x \in \Key$ and $v \in \Val$.
		      By \code{\ref{alg:pc-si-ser}}{\ref{line:readtxn-store-pcsiser}},
		      $v$ is the value written to $x$ by the last transaction,
		      according to $\AR$, in
		      \[
			      V(r) \triangleq \set{T' \in \T \mid T'.\cts < T.\sts}.
		      \]
		      Let $S \triangleq \max_{\AR}(\VIS^{-1}(T) \cap \WriteTx_{x})$.
		      We distinguish three cases depending on whether
		      $T$ is at $\PC$, $\SI$, or $\SER$ level,
		      and show that $S \vdash \W(x, v)$.

		      \begin{itemize}
			      \item \casei: $\level(T) = \PC$.
			            By definition of $\VIS$ for $\PC$ transactions,
			            \[
				            \VIS^{-1}(T) = \set{T' \in \T \mid T'.\cts < T.\sts}.
			            \]
			            Since $V(r) = \VIS^{-1}(T)$, by definition of $S$,
			            \[
				            S = \max_{\AR}(V(r) \cap \WriteTx_{x}).
			            \]
			            That is, $S$ is the last transaction in $V(r)$
			            that writes to $x$ according to $\AR$.
			            Therefore, $S \vdash \W(x, v)$ holds.
			      \item \caseii: $\level(T) = \SI$.
			            The proof is the same as that of \casei{} for $\PC$ transactions.
			      \item \caseiii: $\level(T) = \SER$.
			            By definition of $\VIS$ for $\SER$ transactions,
			            \[
				            \VIS^{-1}(T) = \set{T' \in \T \mid T'.\cts < T.\cts}.
			            \]
			            In the following,	we show that
			            \[
				            V(r) \cap \WriteTx_{x} = \VIS^{-1}(T) \cap \WriteTx_{x}.
			            \]
			            On the one hand, since $T.\sts < T.\cts$,
			            \[
				            V(r) \subseteq \VIS^{-1}(T).
			            \]
			            Hence,
			            \[
				            V(r) \cap \WriteTx_{x} \subseteq
				            \VIS^{-1}(T) \cap \WriteTx_{x}.
			            \]
			            On the other hand, suppose by contradiction that
			            \[
				            \VIS^{-1}(T) \cap \WriteTx_{x} \nsubseteq V(r) \cap \WriteTx_{x}.
			            \]
			            Then there exists a transaction $T' \in \T$ such that
			            \[
				            T' \in \WriteTx_{x} \land T'.\cts \in (T.\sts, T.\cts).
			            \]
			            Since $T \vdash \R(x, v)$, we have $T \rwconflict T'$.
			            By \code{\ref{alg:pc-si-ser}}{\ref{line:committxn-ser-conflict-pcsiser}},
			            $T$ would abort, which is a contradiction.
			            Therefore,
			            \[
				            \VIS^{-1}(T) \cap \WriteTx_{x} \subseteq
				            V(r) \cap \WriteTx_{x}.
			            \]
		      \end{itemize}
		\item $\Session(T)$:
		      Consider any transaction $T' \in \T$
		      such that $T' \rel{\SO} T$.
		      By definition of $\SO$, $T$ starts after $T'$ commits.
		      \begin{itemize}
			      \item \casei: $\level(T) = \PC$.
			            Since $T$ starts after $T'$ commits
			            and the wall time always increases, we have
			            \[
				            T'.\cts < T.\sts.
			            \]
			            By definition of $\VIS$ for $\PC$ transactions,
			            \[
				            T' \rel{\VIS} T.
			            \]
			      \item \caseii: $\level(T) = \SI$.
			            The proof is the same as that of \casei{} for $\PC$ transactions.
			      \item \caseiii: $\level(T) = \SER$.
			            Since $T$ starts after $T'$ commits,
			            $T$ commits after it starts,
			            and the wall time always increases, we have
			            \[
				            T'.\cts < T.\cts.
			            \]
			            By definition of $\VIS$ for $\SER$ transactions,
			            \[
				            T' \rel{\VIS} T.
			            \]
		      \end{itemize}
		\item $\Prefixaxiom(T)$:
		      We need to show that for any transactions $T', S \in \T$,
		      the following holds:
		      \[
			      T' \rel{\AR} S \rel{\VIS} T \implies T' \rel{\VIS} T.
		      \]
		      Since $T' \rel{\AR} S$, by definition of $\AR$,
		      \[
			      T'.\cts < S.\cts.
		      \]
		      In the following, we consider three cases
		      depending on whether $T$ is at $\PC$, $\SI$, or $\SER$ level.
		      \begin{itemize}
			      \item \casei: $\level(T) = \PC$.
			            Since $S \rel{\VIS} T$,
			            by definition of $\VIS$ for $\PC$ transactions,
			            \[
				            S.\cts < T.\sts.
			            \]
			            Thus,
			            \[
				            T'.\cts < S.\cts < T.\sts.
			            \]
			            Therefore, by definition of $\VIS$ for $\PC$ transactions,
			            \[
				            T' \rel{\VIS} T.
			            \]
			      \item \caseii: $\level(T) = \SI$.
			            The proof is the same as that of \casei{} for $\PC$ transactions.
			      \item \caseiii: $\level(T) = \SER$.
			            Since $S \rel{\VIS} T$,
			            by definition of $\VIS$ for $\SER$ transactions,
			            \[
				            S.\cts < T.\cts.
			            \]
			            Thus,
			            \[
				            T'.\cts < S.\cts < T.\cts.
			            \]
			            Therefore, by definition of $\VIS$ for $\SER$ transactions,
			            \[
				            T' \rel{\VIS} T.
			            \]
		      \end{itemize}
		\item $\NoConflict(T)$:
		      Note that in this case, $\level(T) = \SI$
		      or $\level(T) = \SER$.
		      Consider a transaction $T' \neq T$ such that $T \wwconflict T'$.
		      Suppose that $T' \rel{\AR} T$.
		      By definition of $\AR$,
		      \[
			      T'.\cts < T.\cts.
		      \]
		      In the following, we consider two cases
		      depending on whether $T$ is at $\SI$ or $\SER$ level,
		      and show that $T' \rel{\VIS} T$.
		      \begin{itemize}
			      \item \casei: $\level(T) = \SI$.
			            If $T.\sts < T'.\cts$, then we have
			            \[
				            T'.\cts \in (T.\sts, T.\cts) \land T' \wwconflict T.
			            \]
			            By \code{\ref{alg:pc-si-ser}}{\ref{line:committxn-si-conflict-pcsiser}},
			            $T$ should be aborted.
			            Therefore,
			            \[
				            T'.\cts < T.\sts.
			            \]
			            By definition of $\VIS$ for $\SI$ transactions,
			            \[
				            T' \rel{\VIS} T.
			            \]
			      \item \caseii: $\level(T) = \SER$.
			            By definition of $\VIS$ for $\SER$ transactions,
			            \[
				            T' \rel{\VIS} T.
			            \]
		      \end{itemize}
		\item $\TotalVis$:
		      Note that in this case, $\level(T) = \SER$.
		      Consider a transaction $T' \in \T$ such that
		      \[
			      T' \rel{\AR} T.
		      \]
		      By definition of $\AR$,
		      \[
			      T'.\cts < T.\cts.
		      \]
		      By definition of $\VIS$ for $\SER$ transactions,
		      \[
			      T' \rel{\VIS} T.
		      \]
	\end{itemize}
\end{proof}


\section{Proof of Theorem~\ref{thm:si-ser-fekete} (SI--S2PL)}
\label{app-proof-si-ser-fekete}


\begin{proof}[Proof of Theorem~\ref{thm:si-ser-fekete}]
	For any history $\H = (\T, \SO, \level)$ produced by Algorithm~\ref{alg:si-ser-fekete},
	we define $\AR$ as follows:
	\[
		\forall T', T \in \T.\;
			T' \rel{\AR} T \iff T'.\cts < T.\cts.
	\]
	Clearly, $\AR$ is a strict total order.
	The visibility relation $\VIS$ is defined as
	\begin{align*}
		\forall T', T &\in \T.\;
			T' \rel{\VIS} T \iff                                                       \\
				& (\level(T) = \SI \land T'.\cts < T.\sts) \;\lor \\
				& (\level(T) = \SER \land T'.\cts < T.\cts).
	\end{align*}
	Note that we choose $T'.\cts < T.\cts$ for $\SER$ transactions
	rather than $T'.\cts < T.\sts$ (as for $\PC$ and $\SI$ transactions;
		see \code{\ref{alg:pc-si-ser}}{\ref{line:readtxn-store-pcsiser}})
	to ensure that $\SER$ transactions observe all its $\AR$-predecessors,
	required by the $\TotalVis$ axiom.

	It is easy to verify that $\VIS$ is irreflexive
	and $\VIS \subseteq \AR$:
	\begin{itemize}
		\item $\VIS$ is irreflexive:
			For any transaction $T \in \T$,
			by definition of $\VIS$,
			$T \rel{\VIS} T$ implies either
			\[
				\level(T) = \SI \land T.\cts < T.\sts,
			\]
			or
			\[
				\level(T) = \SER \land T.\cts < T.\cts,
			\]
			which is impossible in both cases.
		\item $\VIS \subseteq \AR$:
			Consider any transactions $T', T \in \T$
			such that $T' \rel{\VIS} T$.
			By definition of $\VIS$,
			we have either
			\[
				\level(T) = \SI \land T'.\cts < T.\sts,
			\]
			or
			\[
				\level(T) = \SER \land T'.\cts < T.\cts.
			\]
			In both cases, we have (due to $T.\sts < T.\cts$)
			\[
				T'.\cts < T.\cts.
			\]
			Therefore, by definition of $\AR$,
			\[
				T' \rel{\AR} T.
			\]
	\end{itemize}
	In the following, we show that the corresponding abstract execution
	$\AE = (\H, \VIS, \AR)$ for $\H$ is consistent:
	Fix a transaction $T \in \T$.
	We show that the consistency axioms
	corresponding to $\level(T)$ hold on $T$.

	\begin{itemize}
		\item $\Int(T)$:
			If there is no \textit{internal} read in $T$,
			then $\Int(T)$ holds vacuously.
			Otherwise, $\Int(T)$ holds
			because $T$, whether at $\SI$ or $\SER$ level,
			buffers writes in its private $\buffer$
			(\code{\ref{alg:si-ser-fekete}}{\ref{line:writetxn-buffer-si-ser-fekete}})
			and reads from $\buffer$ for internal reads
			(\code{\ref{alg:si-ser-fekete}}{\ref{line:readtxn-buffer-si-ser-fekete}}).
		\item $\Ext(T)$:
			If there is no \textit{external} read in $T$,
			then $\Ext(T)$ holds vacuously.
			Otherwise, consider any external read $r \triangleq \R(x, v)$ in $T$
			for some $x \in \Key$ and $v \in \Val$.
			Let $S \triangleq \max_{\AR}(\VIS^{-1}(T) \cap \WriteTx_{x})$.
			We distinguish two cases depending on whether
			$T$ is at $\SI$ or $\SER$ level,
			and show that $S \vdash \W(x, v)$.
			\begin{itemize}
				\item \casei: $\level(T) = \SI$.
					By definition of $\VIS$ for $\SI$ transactions,
					\[
						\VIS^{-1}(T) = \set{T' \in \T \mid T'.\cts < T.\sts}.
					\]
					By \code{\ref{alg:si-ser-fekete}}{\ref{line:readtxn-si-store-si-ser-fekete}},
					$v$ is the value written to $x$ by the last transaction,
					according to $\AR$, in
					\[
						V(r) \triangleq \set{T' \in \T \mid T'.\cts < T.\sts},
					\]
					Since $V(r) = \VIS^{-1}(T)$, by definition of $S$,
					\[
						S = \max_{\AR}(V(r) \cap \WriteTx_{x}).
					\]
					That is, $S$ is the last transaction in $V(r)$
					that writes to $x$ according to $\AR$.
					Therefore, $S \vdash \W(x, v)$ holds.
				\item \caseii: $\level(T) = \SER$.
					By definition of $\VIS$ for $\SER$ transactions,
					\[
						\VIS^{-1}(T) = \set{T' \in \T \mid T'.\cts < T.\cts}.
					\]
					By \code{\ref{alg:si-ser-fekete}}{\ref{line:readtxn-ser-store-si-ser-fekete}},
					$v$ is the value written to $x$ by the last transaction,
					according to $\AR$, in
					\[
						V(r) \triangleq \set{T' \in \T \mid T'.\cts < r.\now},
					\]
					where $r.\now$ is the time when $r$ is performed
					(\code{\ref{alg:si-ser-fekete}}{\ref{line:readtxn-ser-store-si-ser-fekete}}).
					In the following,	we show that
					\[
						V(r) \cap \WriteTx_{x} = \VIS^{-1}(T) \cap \WriteTx_{x}.
					\]
					On the one hand, since $r.\now < T.\cts$,
					\[
						V(r) \subseteq \VIS^{-1}(T).
					\]
					Hence,
					\[
						V(r) \cap \WriteTx_{x} \subseteq
						\VIS^{-1}(T) \cap \WriteTx_{x}.
					\]
					On the other hand, suppose by contradiction that
					\[
						\VIS^{-1}(T) \cap \WriteTx_{x} \nsubseteq V(r) \cap \WriteTx_{x}.
					\]
					Then there exists a transaction $T' \in \T$ such that
					\[
						T' \in \WriteTx_{x} \land T'.\cts \in (r.\now, T.\cts).
					\]
					By \code{\ref{alg:si-ser-fekete}}{\ref{line:readtxn-ser-slock}},
					$T$ acquires a shared lock on $x$ before $r.\now$.
					By \code{\ref{alg:si-ser-fekete}}{\ref{line:committxn-si-xlock}}
					for $\SI$ transactions or
					\code{\ref{alg:si-ser-fekete}}{\ref{line:writetxn-ser-xlock}}
					for $\SER$ transactions,
					$T'$ acquires an exclusive lock on $x$.
					Furthermore, both $T$ and $T'$ hold these locks until they commits
					(\code{\ref{alg:si-ser-fekete}}{\ref{line:committxn-releaselocks}}),
					particularly after they are assigned commit timestamps
					(\code{\ref{alg:si-ser-fekete}}{\ref{line:committxn-committime-si-ser-fekete}}).
					Since shared and exclusive locks on the same key are incompatible,
					it cannot be that $T'.\cts \in (r.\now, T.\cts)$.
					Therefore,
					\[
						\VIS^{-1}(T) \cap \WriteTx_{x} \subseteq
							V(r) \cap \WriteTx_{x}.
					\]
			\end{itemize}
		\item $\Session(T)$:
			Consider any transaction $T' \in \T$
			such that $T' \rel{\SO} T$.
			By definition of $\SO$, $T$ starts after $T'$ commits.
			\begin{itemize}
				\item \casei: $\level(T) = \SI$.
					Since $T$ starts after $T'$ commits
					and the wall time always increases, we have
					\[
						T'.\cts < T.\sts.
					\]
					By definition of $\VIS$ for $\SI$ transactions,
					\[
						T' \rel{\VIS} T.
					\]
				\item \caseii: $\level(T) = \SER$.
					Since $T$ starts after $T'$ commits,
					$T$ commits after it starts,
					and the wall time always increases, we have
					\[
						T'.\cts < T.\cts.
					\]
					By definition of $\VIS$ for $\SER$ transactions,
					\[
						T' \rel{\VIS} T.
					\]
			\end{itemize}
		\item $\Prefixaxiom(T)$:
			We need to show that for any transactions $T', S \in \T$,
			the following holds:
			\[
				T' \rel{\AR} S \rel{\VIS} T
					\implies T' \rel{\VIS} T.
			\]
			Since $T' \rel{\AR} S$, by definition of $\AR$,
			\[
				T'.\cts < S.\cts.
			\]
			In the following, we consider two cases
			depending on whether $T$ is at $\SI$ or $\SER$ level.
			\begin{itemize}
				\item \casei: $\level(T) = \SI$.
					Since $S \rel{\VIS} T$,
					by definition of $\VIS$ for $\SI$ transactions,
					\[
						S.\cts < T.\sts.
					\]
					Thus,
					\[
						T'.\cts < S.\cts < T.\sts.
					\]
					Therefore, by definition of $\VIS$ for $\SI$ transactions,
					\[
						T' \rel{\VIS} T.
					\]
				\item \caseii: $\level(T) = \SER$.
					Since $S \rel{\VIS} T$,
					by definition of $\VIS$ for $\SER$ transactions,
					\[
						S.\cts < T.\cts.
					\]
					Thus,
					\[
						T'.\cts < S.\cts < T.\cts.
					\]
					Therefore, by definition of $\VIS$ for $\SER$ transactions,
					\[
						T' \rel{\VIS} T.
					\]
			\end{itemize}
		\item $\NoConflict(T)$:
			Consider a transaction $T' \neq T$ such that $T \wwconflict T'$.
			Suppose that $T' \rel{\AR} T$.
			By definition of $\AR$,
			\[
				T'.\cts < T.\cts.
			\]
			In the following, we consider two cases
			depending on whether $T$ is at $\SI$ or $\SER$ level,
			and show that $T' \rel{\VIS} T$.
			\begin{itemize}
				\item \casei: $\level(T) = \SI$.
					If $T.\sts < T'.\cts$, then we have
					\[
						T'.\cts \in (T.\sts, T.\cts) \land T' \wwconflict T.
					\]
					By \code{\ref{alg:si-ser-fekete}}{\ref{line:committxn-write-conflict-checking-si-ser-fekete}},
					$T$ should be aborted.
					Therefore,
					\[
					  T'.\cts < T.\sts.
					\]
					By definition of $\VIS$ for $\SI$ transactions,
					\[
						T' \rel{\VIS} T.
					\]
				\item \caseii: $\level(T) = \SER$.
					By definition of $\VIS$ for $\SER$ transactions,
					\[
						T' \rel{\VIS} T.
					\]
			\end{itemize}
		\item $\TotalVis(T)$:
			Note that in this case, $\level(T) = \SER$.
		  Consider a transaction $T' \in \T$ such that
			\[
				T' \rel{\AR} T.
			\]
			By definition of $\AR$,
			\[
				T'.\cts < T.\cts.
			\]
			By definition of $\VIS$ for $\SER$ transactions,
			\[
				T' \rel{\VIS} T.
			\]
	\end{itemize}
\end{proof}


\begin{example}
	\label{ex:si-ser-fekete-ce}

	We show that if the procedure $\Call{Commit}{T}$ was not executed atomically
	and \code{\ref{alg:si-ser-fekete}}{\ref{line:committxn-committime-si-ser-fekete}}
	was moved before \code{\ref{alg:si-ser-fekete}}{\ref{line:committxn-si-buffer}},
	then Non-RepeatableRead anomaly may arise.

	Consider the scenario where two transactions $T$ and $T'$
	with $\level(T) = \level(T') = \SI$
	operate on the same key $k$ with initial value $v_0$;
	$T$ is a writer while $T'$ is a reader.

	\begin{enumerate}[leftmargin=15pt]
		\item {$T$ executes \textsc{Commit}.}
		  Because the commit procedure assigns $T.\cts$
			before acquiring locks (in the modified algorithm),
			$T$ immediately sets its commit timestamp to $t$.
			The write to $k$ has not been installed into $\store$ yet.
		\item {$T'$ begins after $t$.}
		  $T'$ receives a start timestamp $s'$ with $s' > t$
			and performs its first read of $k$.
			Since $T$'s write has not yet been installed,
			$T'$ reads the old version (value $v_{0}$).
		\item {$T$ acquires $\xlock(k)$ in \textsc{Commit}.}
			$T$ continues its commit procedure,
			acquiring an exclusive lock on $k$
			and installing the new value $v_{t}$ into $\store$
			with its commit time $t$.
		\item {$T'$ reads $k$ again.}
			$T'$ now reads the current version of $k$ and sees $v_{t}$.
			Thus, $T'$ has read two different values ($v_{0}$ and $v_{t}$)
			for the same key in the same transaction,
			violating the repeatable read guarantee.
	\end{enumerate}
\end{example}



\end{document}